\documentclass[11pt]{article}

\usepackage{multirow}
\usepackage{enumitem}
\usepackage{caption}
\usepackage[colorlinks, linkcolor=blue, anchorcolor=red, citecolor=blue, CJKbookmarks=true]{hyperref}
\usepackage{epsfig, epstopdf, graphicx, floatrow, rotating, amssymb, bbm, amsthm, amsbsy, amsmath, amsfonts, mathrsfs, dsfont, makeidx, color}
\usepackage{bm}
\usepackage{algorithmicx}
\usepackage{algpseudocode}
\usepackage{natbib}
\usepackage{graphicx}
\usepackage{float}
\usepackage{subfigure}
\usepackage{xr}
\usepackage[strict]{changepage}

\usepackage[left=4cm,right=2.5cm,top=2cm,bottom=4cm,includehead,includefoot,headheight=15pt]{geometry}
\allowdisplaybreaks

\usepackage{setspace}
\graphicspath{{images/}}
\numberwithin{equation}{section}

\theoremstyle{plain}
\newtheorem{theorem}{Theorem}

\newtheorem{lemma}{Lemma}[section]
\newtheorem{claim}{Claim}[section]

\newtheorem{assumption}{Assumption}[section]

\theoremstyle{remark}
\newtheorem*{remark}{Remark}
\newtheorem{example}{Example}[section]

\DeclareMathOperator*{\argmax}{argmax}

\DeclareMathOperator*{\diag}{diag}
\DeclareMathOperator*{\var}{Var}

\newcommand\norm[1]{\left\lVert#1\right\rVert}
\newcommand\abs[1]{\left|#1\right|}
\newcommand\bracketS[1]{\left(#1\right)}
\newcommand\bracketM[1]{\left[#1\right]}
\newcommand\bracketL[1]{\left\{#1\right\}}

\newcommand{\expct}{\mathbb{E}}
\newcommand{\prob}{\mathbb{P}}

\newcommand{\cH}{\mathcal{H}}

\newcommand{\ve}{\mathbf{e}}
\newcommand{\vh}{\mathbf{h}}

\newcommand{\vv}{\mathbf{v}}

\newcommand{\vx}{\mathbf{x}}
\newcommand{\vy}{\mathbf{y}}

\newcommand{\vone}{\mathbf{1}}

\newcommand{\bL}{\mathbb{L}}
\newcommand{\bN}{\mathbb{N}}
\newcommand{\bR}{\mathbb{R}}

\newcommand{\bY}{\mathbb{Y}}
\newcommand{\bZ}{\mathbb{Z}}

\newcommand{\Sighat}{\widehat{\Sigma}}
\newcommand{\convas}{\overset{a.s.}{\to}}
\newcommand{\convp}{\overset{p}{\to}}
\newcommand{\convd}{\overset{d}{\to}}

\newcommand{\bigO}{\mathcal{O}}

\begin{document}

\title{Modelling Volatility of Spatio-temporal Integer-valued Data\\
with Network Structure and Asymmetry 
}
\author{Yue Pan \quad and \quad Jiazhu Pan  \\
\textsl {\footnotesize Department of Mathematics and Statistics, University of Strathclyde, Glasgow G1 1XH, UK }
}
\date{}
\maketitle

\begin{abstract}
This paper proposes a spatial threshold GARCH-type model for dynamic spatio-temporal integer-valued data with network structure. The proposed model can simplify the parameterization by using network structure in data, and can capture the asymmetric property in dynamic volatility by adopting a threshold structure. The proposed model assumes the conditional distribution is Poisson distribution. Asymptotic theory of maximum likelihood estimation (MLE) for the spatial model is derived when both sample size and network dimension are large. We obtain asymptotic statistical inferences via investigation of the weak dependence of components of the model and application of limit theorems for weakly dependent random fields. Simulation studies and a real data example are presented to support our methodology.
\end{abstract}

{\small \textbf{Keywords}: Conditional variance/mean, high-dimensional count time series, asymmetry of volatility, spatial threshold GARCH, network structure.\\

{\small\textbf{MSC 2020 Subject Classification}: Primary 62M10, 91B05; Secondary 60G60, 60F05.}


\section{Introduction}

Integer-valued time series can be observed in a wide range of scientific fields, such as the yearly trading volume of houses on real estate market \citep{de2013}, number of transactions of stocks \citep{jones1994}, daily mortality from Covid-19 \citep{pham2020} and daily new cases of Covid-19 \citep{jiang2022}. A first idea to model integer-valued time series is using a simple first-order autoregressive model (AR)
\begin{equation}\label{model_ar}
    X_t = \alpha X_{t-1} + \varepsilon_t,
\end{equation} 
where $0\leq\alpha<1$ is a parameter. However in \eqref{model_ar} $X_t$ is not necessarily an integer given integer-valued $X_{t-1}$ and $\varepsilon_t$, due to the multiplication structure $\alpha X_{t-1}$. To circumvent such problem,  \cite{mckenzie1985} and \cite{al1987} proposed an integer-valued counterpart of AR model (INAR). They replaced the ordinary multiplication $\alpha X_{t-1}$ by the binomial thinning operation $\alpha\circ X_{t-1}$, where $\alpha\circ X|X \sim Bin(X, \alpha)$. This was ground-breaking and led to various extensions of thinning-based linear models including integer-valued moving average model (INMA) \citep{al1988} and INARMA model \citep{mckenzie1988}. An alternative approach to the multiplication problem is to consider the regression of the conditional mean $\lambda_t := \expct(X_t|\cH_{t-1})$, where $H_{t-1}$ is the $\sigma$-algebra generated by historical information up to $t-1$. Based on this idea, integer-valued GARCH-type models (INGARCH) were proposed by \cite{heinen2003}, \cite{ferland2006} and \cite{fokianos2009} with conditional Poisson distribution of $X_t$, i.e.
\begin{equation}\label{model_pgarch}
\begin{aligned}
    &X_t|\cH_{t-1} \sim Poisson(\lambda_t),\\
    &\lambda_t = \omega + \sum_{i=1}^p\alpha_i X_{t-i} + \sum_{j=1}^q\beta_j \lambda_{t-j},   
\end{aligned}
\end{equation}
where  parameters satisfy $\omega>0, \alpha_i \geq 0,   i=1,\cdots, p, \beta_j \geq 0,    j=1,\cdots, q$.  Other variations of INGARCH models with different specifications of conditional distribution include negative binomial INGARCH \citep{zhu2010, xu2012} and generalized Poisson INGARCH \citep{zhu2012}.

 Integer-valued models mentioned above are all limited to one-dimensional time series.  The development of multivariate integer-valued GARCH-type models is still at its early stage. For example, the bivariate INGARCH models \citep{lee2018, cui2018, cui2020}, multivariate INGARCH models \citep{fokianos2020, lee2023} and some counterparts of the continuous-valued network GARCH model \citep{zhou2020} are on fixed low-dimensional time series of counts. For high-dimensional integer-valued time series with increasing dimension, there is the Poisson network autoregressive model (PNAR) by \cite{armillotta2024}.  The PNAR allows for integer-valued time series with increasing network dimension. However, it adopted a ARCH-type structure without considering the autoregressive term on the conditional mean/variance to describe persistence in volatility, and it can not capture asymmetric characteristics of volatility. \cite{tao2024} proposed a grouped PNAR model which has a GARCH structure, but its network dimension is fixed and not suitable to spatio-temporal integer-valued data. In this paper, we propose a Poisson threshold network GARCH model (PTNGARCH) and discuss its asymptotic inferences. The specific contributions are as follows.  
\begin{enumerate}
\item A threshold structure is designed so that it is able to capture asymmetric property of high-dimensional volatility for integer-valued data. The threshold effect can also be tested under such a framework.  
\item 
Our PTNGARCH includes an autoregressive term on the conditional mean/variance so that it provides a parsimonious description of dynamic volatility persistence of high-dimensional integer-valued time series. 
\item 
Asymptotic theory, when both sample size ($T$) and network dimension ($N$) are large, of maximum likelihood estimation for the proposed model is established by the limit theorems for arrays of weakly dependent random fields in \cite{pan2024ltrf}. This enables the proposed model to be suitable for fitting dynamic spatio-temporal integer-valued data (or the so-called ultra-high dimensional integer-valued time series).
\end{enumerate}

 The remainder of this paper is arranged as follows. The PTNGARCH model is introduced and its stationarity over time is discussed under fixed network dimension in section \ref{section_model}. Section \ref{section_inference} presents the MLE for parameters, including the threshold, and the consistency and asymptotic normality for estimates of coefficients under large sample size ($T$) and large network dimension ($N$). Section \ref{section_inference_testing} gives a Wald test which can be applied to testing of the existence of threshold effect (i.e. asymmetric property), GARCH effect and network effect. Section \ref{section_numerical_study} conducts simulation studies to verify the asymptotic properties of the MLE, and applies the proposed model to daily numbers of car accidents that occurred in 41 neighbourhoods in New York City, with interpretation of the results of analysis. All proofs of theoretical results are postponed to the appendix.

\section{Model setting and stationarity over time}\label{section_model}

Consider an non-directed and weightless network with $N$ nodes. Define adjacency matrix $A = (a_{ij})_{1\leq i,j\leq N}$, where $a_{ij} = 1$ if there is a connection between node $i$ and $j$, otherwise $a_{ij} = 0$. In addition, self-connection is not allowed for any node $i$ by letting $a_{ii} = 0$. As an interpretation of the network structure, $A$ is symmetric since $a_{ij} = a_{ji}$, hence for any node $i$, the \textbf{out-degree} $d_i^{(out)} = \sum_{j=1}^N a_{ij}$ is equal to the \textbf{in-degree} $d_i^{(in)} = \sum_{j=1}^N a_{ji}$, and we use $d_i$ to denote both for convenience. To embed a network into statistical models, it is often convenient to use the row normalized adjacency matrix $W$ with its $(i,j)$ element $w_{ij} = \frac{a_{ij}}{d_i}$.

For any node $i\in\{1,\cdots.N\}$ in this network, let $y_{it}$ be an non-negative integer-valued observation at time $t$, and $\mathcal{H}_{t-1}$ denotes the $\sigma$-algebra consisting of all available information up to $t-1$. In our Poisson threshold network GARCH model, for each $i = 1, 2, ..., N$ and $t \in \bZ$, $y_{it}$ is assumed to follow a conditional (on $\cH_{t-1}$) Poisson distribution with $(i,t)$-varying variance (mean) $\lambda_{it}$. A PTNGARCH(1,1) model has the following form:
\begin{equation}\label{model_ptngarch}
\begin{aligned}
    &y_{it}|\cH_{t-1} \sim Poisson(\lambda_{it}),\\
    &\lambda_{it} = \omega + \left(\alpha^{(1)}1_{\{y_{i,t-1} \geq r\}} + \alpha^{(2)}1_{\{y_{i,t-1} < r\}}\right) y_{i,t-1} + \xi\sum_{j=1}^N w_{ij}  y_{j,t-1} + \beta\lambda_{i,t-1},\\
    & i=1,2,\cdots, N.
\end{aligned}
\end{equation} 
The threshold parameter $r$ is a positive integer, and $1_{\{\cdot\}}$ denotes an indicator function. To assure the positiveness of conditional variance, we need to assume positiveness of the base parameter $\omega$, and non-negativeness of all the coefficients $\alpha^{(1)}$, $\alpha^{(2)}$, $\xi$, $\beta$.
\begin{remark}
    Notice that in \eqref{model_ptngarch} we model the dynamics of conditional mean $\lambda_{it}$, which is the reason why the name ``Poisson autoregression" is sometimes used in the literature \citep{fokianos2009, wang2014}; Some authors still use the name ``GARCH" since the mean is equal to the variance under Poisson distribution, and the dynamics of conditional mean are GARCH-like.
\end{remark}

Let $\{N_{it}: i = 1, 2, ..., N, t\in\bZ\}$ be independent Poisson processes with unit intensities. Depending on $\lambda_{it}$, $y_{it}$ can be interpreted as a Poisson distributed random variable $N_{it}(\lambda_{it})$, which is the number of occurrences during the time interval $(0, \lambda_{it}]$, i.e. $\prob(y_{it} = n|\lambda_{it} = \lambda) = \frac{\lambda^n}{n!}e^{-\lambda}.$ We could rewrite \eqref{model_ptngarch} in a vectorized form as follows:
\begin{equation}\label{model_ptngarch2}
    \begin{cases}
         &\bY_t = (N_{1t}(\lambda_{1t}), N_{2t}(\lambda_{2t}), ..., N_{Nt}(\lambda_{Nt}))',\\
         &\Lambda_t = \omega\vone_N + A(\bY_{t-1})\bY_{t-1} + \beta\Lambda_{t-1},
    \end{cases}
\end{equation} where 
\begin{align*}
    &\Lambda_t = (\lambda_{1t}, \lambda_{2t}, ..., \lambda_{Nt})' \in \bR^N,\\
    &\vone_N = (1, 1, ..., 1)' \in \bR^N,\\
    &A(\bY_{t-1}) = \alpha^{(1)} S(\bY_{t-1}) + \alpha^{(2)} (I_N - S(\bY_{t-1})) + \xi W,\\
    &S(\bY_{t-1}) = \diag\left\{1_{\{y_{1,t-1} \geq r\}}, 1_{\{y_{2,t-1} \geq r\}}, ..., 1_{\{y_{N,t-1} \geq r\}}\right\}.
\end{align*}
Note that $\bY_t \in \bN^N$ with dimension $N$ and $\bN=\{0,1,2,\cdots\}$. For a fixed dimension $N$, the stationarity condition of time series $\bY_t$ can be obtained. The proof of the following theorem is given in the appendix.
\begin{theorem}\label{theorem_stationarity} If the parameters satisfy
\begin{equation}\label{assumption_contraction}
\max\left\{\alpha^{(1)}, \alpha^{(2)}, \abs{\alpha^{(1)}r - \alpha^{(2)}(r-1)}\right\} + \xi + \beta < 1,
\end{equation}
there exists a unique strictly stationary process $\{\bY_t: t\in\bZ\}$ that satisfies \eqref{model_ptngarch2} and has finite first order moment. Therefore, each component $y_{it}$ of model \eqref{model_ptngarch} has a unique strictly stationary solution with finite first order moment.
\end{theorem}

\section{Parameter estimation with \texorpdfstring{$T\to\infty$}{} and \texorpdfstring{$N\to\infty$}{}}\label{section_inference}

Denote parameter vector of  model \eqref{model_ptngarch} by $\nu = (\theta', r)'$ with $\theta = (\omega, \alpha^{(1)}, \alpha^{(2)}, \xi, \beta)'$. Let  
$\Theta=\{\theta: \omega>0, \alpha^{(1)}\ge 0, \alpha^{(2)}\ge 0, \xi\ge 0, \beta\ge 0\}$ and $\bZ_{+}=\{0,1,2,\cdots\}$. For a given threshold $r$, the reasonable parameter space for $\theta$ is a sufficiently large compact subset  of $\Theta$, such as $\Theta_{\delta}=\{\theta: \omega\ge \delta, \alpha^{(1)}\ge 0, \alpha^{(2)}\ge 0, \xi\ge 0, \beta\ge 0\}$ with sufficient small positive real number $\delta$.

Let $D_{NT}=\{(i,t): i=1,\cdots,N;t=1,\cdots,T\}$ be the index set. Suppose that the samples $\{y_{it}: (i,t) \in D_{NT}, NT \geq 1\}$ are generated by model \eqref{model_ptngarch} with respect to true parameters $\nu_0 = (\omega_0, \alpha^{(1)}_0, \alpha^{(2)}_0, \xi_0, \beta_0, r_0)'$.

The log-likelihood function (ignoring constants) is
\begin{equation}\label{lik_func_1_ptngarch}
\left\{
    \begin{array}{ll}
         &L_{NT}(\nu) = \frac{1}{NT}\sum_{(i,t)\in D_{NT}} l_{it}(\nu),\\
         &l_{it}(\nu) = y_{it}\log\lambda_{it}(\nu) - \lambda_{it}(\nu)
    \end{array}\right.
\end{equation} 
where $\lambda_{it}(\nu)$ is generated from model \eqref{model_ptngarch} as
\begin{equation}\label{eq_lambda_ptngarch}
\begin{aligned}
    \lambda_{it}(\nu) = &\omega + \alpha^{(1)}1_{\{y_{i,t-1}\geq r\}}y_{i,t-1} + \alpha^{(2)}1_{\{y_{i,t-1}<r\}}y_{i,t-1}\\
    &+ \xi\sum_{j=1}^N w_{ij}  y_{j,t-1} + \beta\lambda_{i,t-1}(\nu).
\end{aligned}
\end{equation} 
In practice, \eqref{lik_func_1_ptngarch} can not be evaluated without knowing the true values of $\lambda_{i0}$ for $i = 1, 2, ..., N$. We need to approximate \eqref{lik_func_1_ptngarch} by \eqref{lik_func_2_ptngarch} below, using specified initial values $\lambda_{i0} = \Tilde{\lambda}_{i0}, i = 1, 2, ..., N$:
\begin{equation}\label{lik_func_2_ptngarch}
\left\{
    \begin{array}{ll}
         &\Tilde{L}_{NT}(\nu) = \frac{1}{NT}\sum_{(i,t)\in D_{NT}} \Tilde{l}_{it}(\nu),\\
         &\Tilde{l}_{it}(\nu) = y_{it}\log\Tilde{\lambda}_{it}(\nu) - \Tilde{\lambda}_{it}(\nu).
    \end{array}\right.
\end{equation} 
Therefore, the maximum likelihood estimate (MLE) of parameter $\nu$ is defined as
\begin{equation}\label{MLE_ptngarch}
    \hat{\nu}_{NT} = \argmax_{\nu \in \Theta\times\bZ_{+}}\Tilde{L}_{NT}(\nu).
\end{equation} 
However, the solution that maximizes the target function $\Tilde{L}_{NT}(\nu)$ can not be directly obtained by solving $\frac{\partial\Tilde{L}_{NT}(\nu)}{\partial\nu} = 0,$ since $r\in\bZ_{+}$ is discrete-valued, therefore the partial derivative of $\Tilde{L}_{NT}(\nu)$ w.r.t. $r$ is invalid. According to \cite{wang2014}, such optimization problem with integer-valued parameter $r$ could be broken up into two steps as follows:
\begin{enumerate}
    \item Find $$\hat{\theta}_{NT}^{(r)} = \argmax_{\theta\in\Theta}\Tilde{L}_{NT}(\theta, r)$$ for each $r$ in a predetermined range $[\underline{r}, \Bar{r}] \subset \bZ_{+}$;
    \item Find $$\hat{r}_{NT} = \argmax_{r \in [\underline{r}, \Bar{r}]}\Tilde{L}_{NT}(\hat{\theta}_{NT}^{(r)}, r).$$
\end{enumerate} Then $\hat{\nu}_{NT} = \left(\hat{\theta}_{NT}^{(\hat{r}_{NT})'}, \hat{r}_{NT}\right)'$ would be the maximizer of $\Tilde{L}_{NT}(\nu)$.

Assumption \ref{assumption_theta_ptngarch} below is a regular condition on the parameter space. Assumptions \ref{assumption_y_ptngarch} and \ref{assumption_network_ptngarch} are necessary for obtaining $\eta$-weak dependence of $\{l_{it}(\nu): (i,t)\in D_{NT}, NT \geq 1\}$ and their derivatives. For the definition of $\eta$-weak dependence for random fields, see \cite{pan2024ltrf}.   Then the consistency of MLE in Theorem \ref{theorem_C_ptngarch} can be proved based on the law of large numbers (LLN) of $\eta$-weakly dependent arrays of random fields in \cite{pan2024ltrf}.

\begin{assumption}\label{assumption_theta_ptngarch}
For a given threshold $r\in \bZ_{+}$,
\begin{itemize}
        \item [(a)] The parameter space for $\theta$ is a compact subset of $\Theta$ and includes the true parameter $\theta_0$ as its interior point; 
        \item [(b)] Condition \eqref{assumption_contraction} holds.
\end{itemize}        
\end{assumption}

\begin{assumption}\label{assumption_y_ptngarch}
    \begin{itemize}
        \item [(a)] $\sup_{NT\geq 1}\sup_{(i,t)\in D_{NT}}\expct|y_{it}|^{2p} < \infty$ for some $p > 1$;
        \item [(b)] The array of random fields $\left\{y_{it}: (i,t)\in D_{NT}, NT\geq 1\right\}$ is $\eta$-weakly dependent with coefficient $\Bar{\eta}_y(s) := \bigO(s^{-\mu_y})$ for some $\mu_y > 2\frac{2p-1}{p-1}$.
    \end{itemize}
\end{assumption}
\begin{assumption}\label{assumption_network_ptngarch}
    For any $i = 1, 2, ..., N$ and $j = 1, 2, ..., N$, there exist constants $C > 0$ and $b > \mu_y$ such that $w_{ij} \leq C|j-i|^{-b}.$ That is, the power of connection between two nodes $i$ and $j$ decays as the distance $|i-j|$ grows.
\end{assumption}

\begin{theorem}\label{theorem_C_ptngarch}
If Assumptions \ref{assumption_theta_ptngarch}, \ref{assumption_y_ptngarch} and \ref{assumption_network_ptngarch} hold, the MLE defined by \eqref{MLE_ptngarch} is consistent, i.e.
$$\hat{\nu}_{NT} \convp \nu_0$$ as $T\to\infty$ and $N\to\infty$.
\end{theorem}

Since $\hat{r}_{NT}$ is an integer-valued consistent estimate of $r_0$, $\hat{r}_{NT}$ will eventually be equal to $r_0$ when the sample size $NT$ becomes sufficiently large. Therefore, $\hat{\nu}_{NT} = \left(\hat{\theta}_{NT}^{(\hat{r}_{NT})'}, \hat{r}_{NT}\right)'$ is asymptotically equal to $\left(\hat{\theta}_{NT}^{(r_0)'}, r_0\right)'$. In this way, the problem of investigating the asymptotic distribution of $\hat{\nu}_{NT}$ degenerates to investigating the asymptotic distribution of $\hat{\theta}_{NT}^{(r_0)}$.

\begin{theorem}\label{theorem_AN_ptngarch}
Assume that all conditions in Theorem \ref{theorem_C_ptngarch} are satisfied with $\mu_y > \frac{6p-3}{p-1}\vee\frac{(4p-3)(2p-1)}{2(p-1)^2}$ in Assumption \ref{assumption_y_ptngarch}(b) instead. If the smallest eigenvalue $\lambda_{min}(\Sigma_{NT})$ of $$\Sigma_{NT} := \frac{1}{NT}\sum_{(i,t)\in D_{NT}} \expct \left[\frac{1}{\lambda_{it}(\nu_0)}\frac{\partial\lambda_{it}(\nu_0)}{\partial\theta}\frac{\partial\lambda_{it}(\nu_0)}{\partial\theta'}\right]$$ satisfies that
\begin{equation}\label{condition_sigma_ptngarch}
    \inf_{NT\geq 1} \lambda_{min}(\Sigma_{NT}) > 0,
\end{equation} then $\hat{\theta}_{NT}^{(r_0)}$ is asymptotically normal, i.e. $$\sqrt{NT}\Sigma_{NT}^{1/2}(\hat{\theta}_{NT}^{(r_0)} - \theta_0) \convd N(0, I_5)$$ as $T\to\infty$, $N\to\infty$ and $N = o(T)$.
\end{theorem}

\begin{remark}
    In the proof of Theorem \ref{proposition_wald_test_ptngarch} below, we will show that, $\Sigma_{NT}$ could be consistently estimated by 
    \begin{equation}\label{hat_sigma_nt}
    \Sighat_{NT} = \frac{1}{NT}\sum_{(i,t)\in D_{NT}}\left[\frac{1}{\Tilde{\lambda}_{it}(\hat{\nu}_{NT})}\frac{\partial\Tilde{\lambda}_{it}(\hat{\nu}_{NT})}{\partial\theta}\frac{\partial\Tilde{\lambda}_{it}(\hat{\nu}_{NT})}{\partial\theta'}\right]
    \end{equation} in practice.
In the simulation studies in Section \ref{section_numerical_study}, the elements of $\Sighat_{NT}$ are used to construct confidence intervals and to evaluate coverage probabilities.     
\end{remark}

\section{Hypothesis testing with \texorpdfstring{$T\to\infty$}{} and \texorpdfstring{$N\to\infty$}{}}\label{section_inference_testing}

Based on Theorem \ref{theorem_C_ptngarch} and Theorem \ref{theorem_AN_ptngarch}, for sufficiently large sample region such that $\hat{r}_{NT} = r_0$, we are able to design a Wald test for the null hypothesis
\begin{equation}\label{null_hypo_wald_ptngarch}
    H_0: \Gamma\theta_0 = \eta,
\end{equation} 
where $\Gamma$ is an $s\times 5$ matrix with rank $s$ ( $1\le s\le 5$) and $\eta$ is an $s$-dimensional vector. For example, to test the existence of threshold effect, simply let $\Gamma = (0, 1, -1, 0, 0)$ and $\eta = 0$, and the null hypothesis \eqref{null_hypo_wald_ptngarch} becomes $$H_0: \alpha^{(1)}_0 = \alpha^{(2)}_0;$$
To see if the autoregressive term is necessary (i.e. GARCH effect), take $\Gamma = (0, 0, 0, 0, 1)$ and $\eta = 0$, and the  hypothesis becomes $$H_0: \beta_0 = 0\quad v.s. \quad H_1: \beta_0 > 0;$$ 
To test the network effect, just let $\Gamma = (0, 0, 0, 1, 0)$ and $\eta = 0$, and the question becomes testing the hypothesis $$H_0: \xi_0 = 0 \quad v.s. \quad H_1: \xi_0 > 0.$$

Corresponding to the asymptotic normality of $\hat{\theta}_{NT}^{(r_0)}$ in Theorem \ref{theorem_AN_ptngarch}, we define a Wald test statistic as follows
\begin{equation}\label{wald_statistic_ptngarch}
    W_{NT} := (\Gamma\hat{\theta}_{NT}^{(r_0)} - \eta)'\left\{\frac{1}{NT}\Gamma\Sighat_{NT}^{-1}\Gamma'\right\}^{-1}(\Gamma\hat{\theta}_{NT}^{(r_0)} - \eta),
\end{equation} 
where $\Sighat_{NT}$ is defined as \eqref{hat_sigma_nt}.
The following Theorem \ref{proposition_wald_test_ptngarch} shows that $W_{NT}$ has an asymptotic $\chi^2$-distribution with degree of freedom $s$.

\begin{theorem}\label{proposition_wald_test_ptngarch}
Under the same assumptions as in Theorem \ref{theorem_AN_ptngarch},  the Wald test statistic defined in \eqref{wald_statistic_ptngarch} asymptotically follows a $\chi^2$ distribution with degree of freedom $s$, i.e.   $$W_{NT}\convd\chi^2_s,$$
as $T\to\infty$, $N\to\infty$ and $N = o(T)$.
\end{theorem}

\section{Simulation studies and empirical data analysis}\label{section_numerical_study}

\subsection{Simulation studies}

We intend to use four different mechanisms of simulating the network structure in model \eqref{model_ptngarch}. The network structure in Example \ref{example_band} satisfies Assumption \ref{assumption_network_ptngarch}. Simulation mechanisms in Examples \ref{example_random} -- \ref{example_block} are for testing the robustness of our estimation, against network structures that may violate Assumption \ref{assumption_network_ptngarch}.

\begin{example}\label{example_band}
\textbf{(D-neighbourhood network)} For each node $i\in\left\{1, 2, ..., N\right\}$, it is connected to node $j$ only if $j$ is inside $i$'s $D$-neighbourhood. That is, in the adjacency matrix, $a_{ij} = 1$ if $0 < |i-j| \leq D$ and $a_{ij} = 0$ otherwise. Figure \ref{fig:visualized_network}(a) is a visualization of such a network with $N = 100$ and $D = 10$.
\end{example}

\begin{example}\label{example_random}
\textbf{(Random network)} For each node $i\in\left\{1, 2, ..., N\right\}$, we generate $D_i$ from uniform distribution $U(0, 5)$, and then draw $[D_i]$ samples randomly from $\left\{1, 2, ..., N\right\}$ to form a set $S_i$ ($[x]$ denotes the integer part of $x$). $A = (a_{ij})$ could be generated by letting $a_{ij} = 1$ if $j\in S_i$ and $a_{ij} = 0$ otherwise. In a network simulated with such mechanism, as it is indicated in Figure \ref{fig:visualized_network}(b), there is no significantly influential node (i.e. node with extremely large in-degree).
\end{example}

\begin{example}\label{example_pwl}
\textbf{(Network with power-law)} According to \cite{clauset2009}, for each node $i$ in such a network, $D_i$ is generated the same way as in Example \ref{example_random}. Instead of uniformly selecting $[D_i]$ samples from $\left\{1, 2, ..., N\right\}$, these samples are collected w.r.t. probability $p_i = s_i/\sum_{i=1}^N s_i$ where $s_i$ is generated from a discrete power-law distribution $\prob\left\{s_i = x\right\} \propto x^{-a}$ with scaling parameter $a = 2.5$. As shown in Figure \ref{fig:visualized_network}(c), a few nodes have much larger in-degrees while most of them have less than 2. Compared to Example \ref{example_random}, network structure with power-law distribution exhibits larger gaps between the influences of different nodes. This type of network is suitable for modeling social media such as Twitter and Instagram, where celebrities have huge influence while the ordinary majority has little.
\end{example}

\begin{example}\label{example_block}
\textbf{(Network with K-blocks)} As it was proposed in \cite{nowicki2001}, in a network with stochastic block structure, all nodes are divided into blocks and nodes from the same block are more likely to be connected comparing to those from different blocks. To simulate such structure, these $N$ nodes are randomly divided into $K$ groups by assigning labels $\left\{1, 2, ..., K\right\}$ to every nodes with equal probability. For any two nodes $i$ and $j$ from the same group, let $\prob(a_{ij} = 1) = 0.5$ while for those two from different groups, $\prob(a_{ij} = 1) = 0.001/N$. Hence, it is very unlikely for nodes to be connected across groups. Our simulated network successfully mimics this characteristic as Figure \ref{fig:visualized_network}(d) shows clear boundaries between groups. Block network also has its advantage in practical perspective. For instance, the price of one stock is highly relevant to those in the same industry sector.
\end{example}

\begin{figure}[htbp]
\centering
    \subfigure[Example \ref{example_band} (D = 10)]{
    \includegraphics[width=0.35\linewidth]{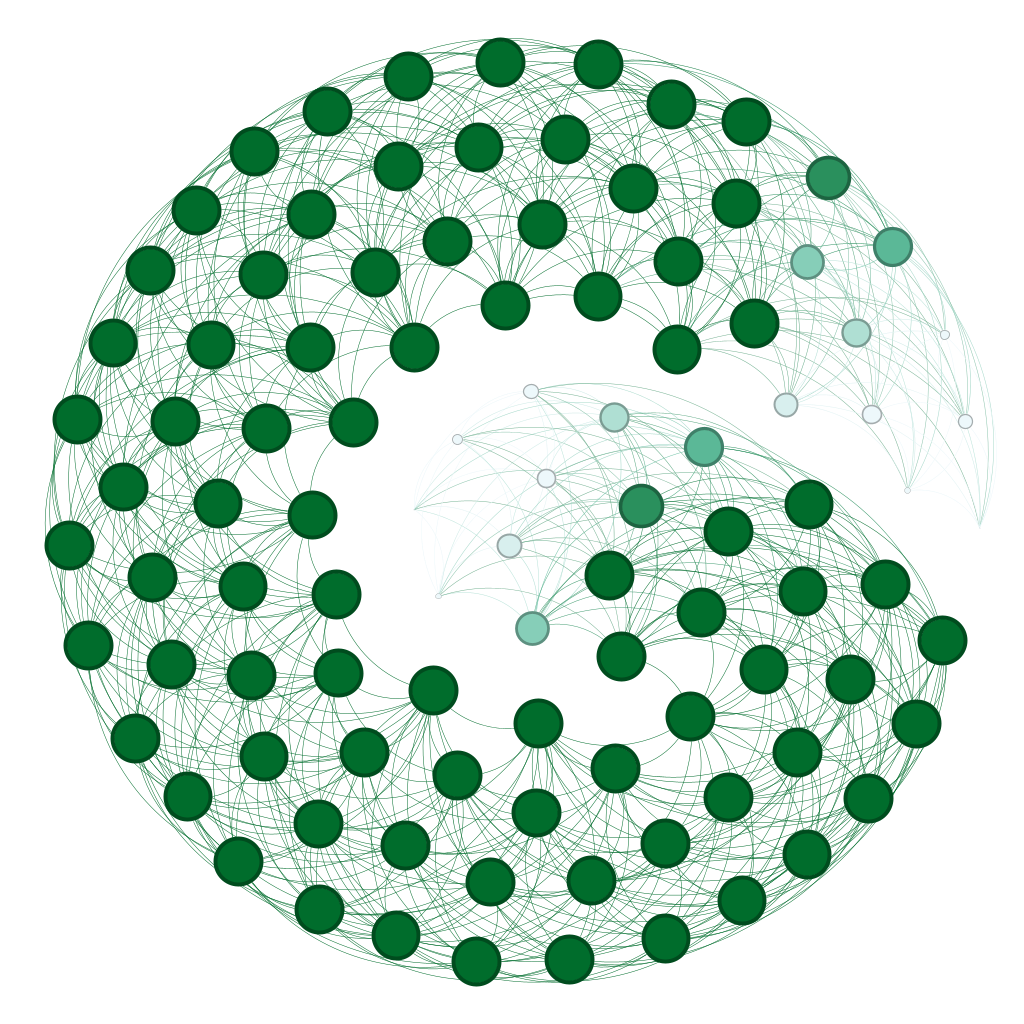}
    }
    \subfigure[Example \ref{example_random}]{
    \includegraphics[width=0.35\linewidth]{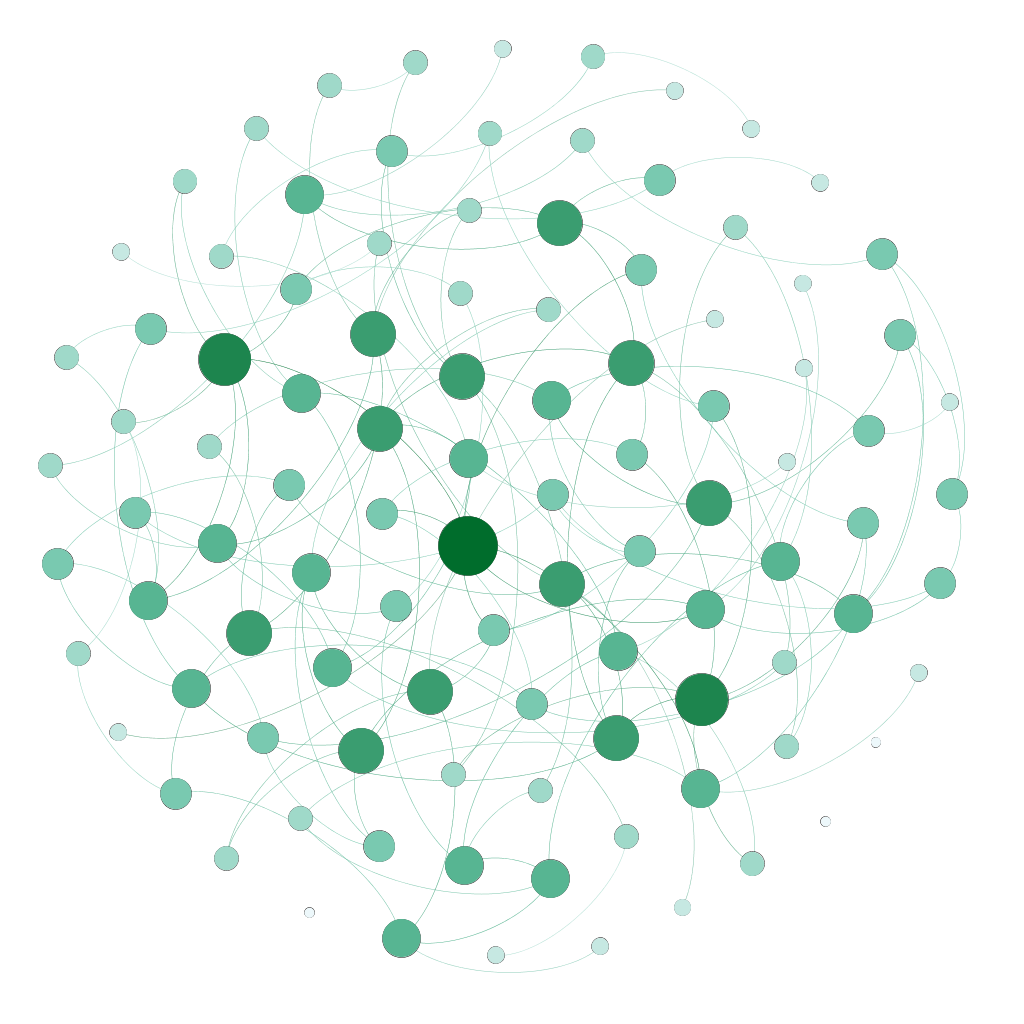}
    }
    \subfigure[Example \ref{example_pwl}]{
    \includegraphics[width=0.35\linewidth]{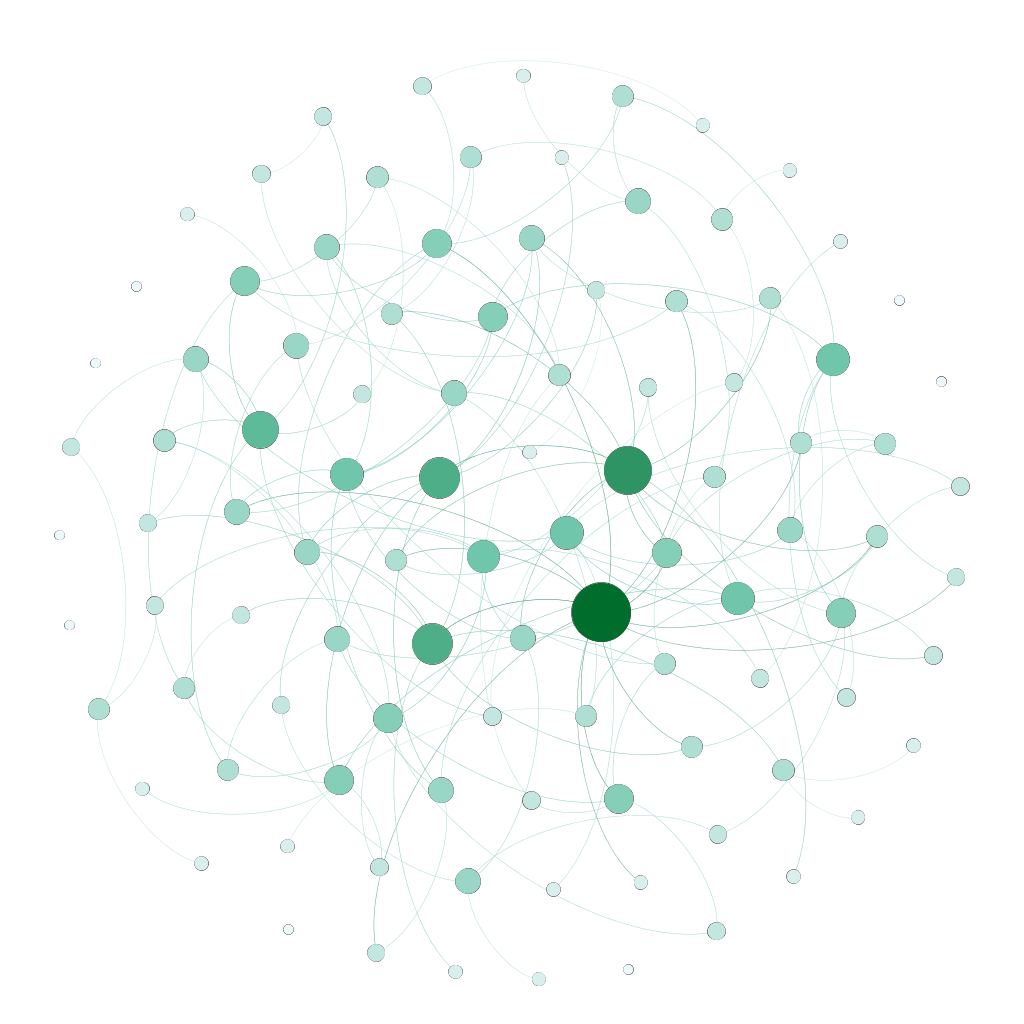}
    }
    \subfigure[Example \ref{example_block} (K = 10)]{
    \includegraphics[width=0.35\linewidth]{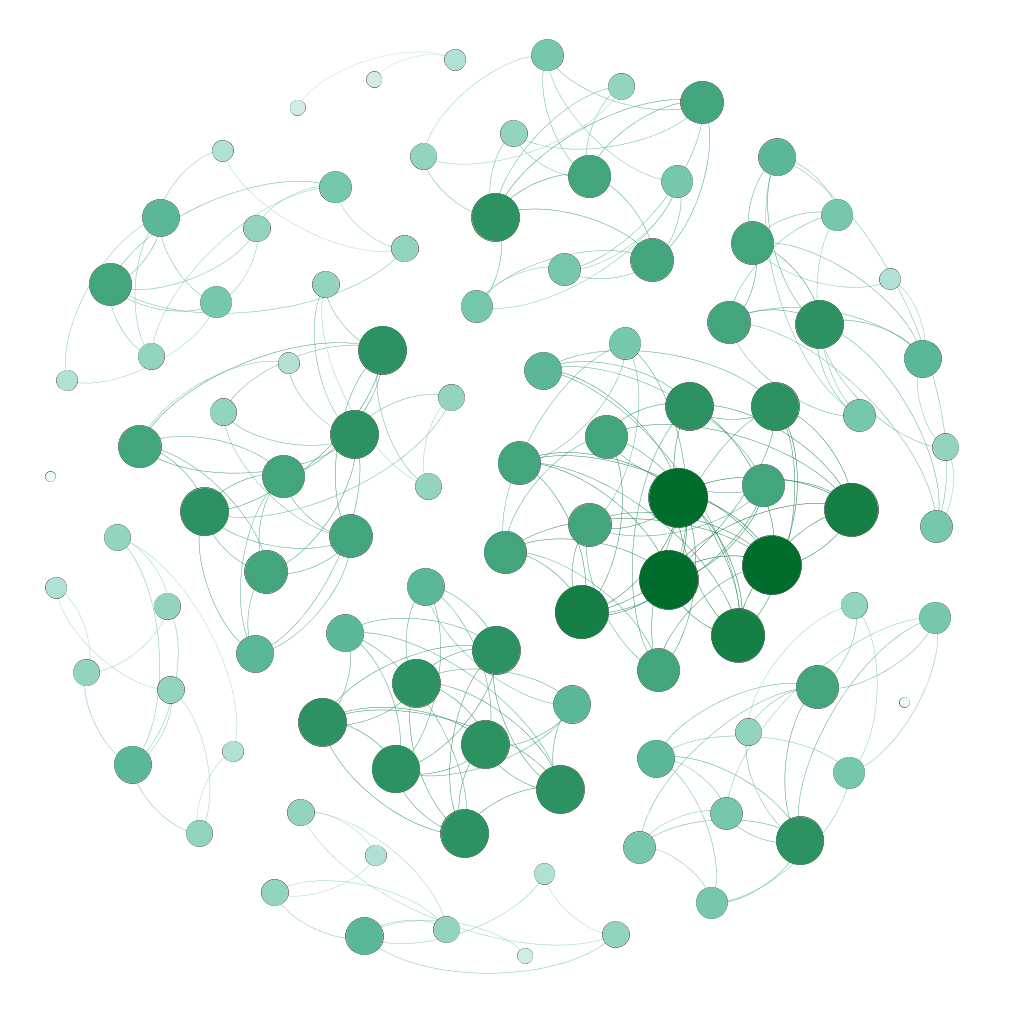}
    }
    \caption{Visualized network structures with N = 100}
    \label{fig:visualized_network}
\end{figure}

Set the true parameters $\nu_0 = (0.5, 0.7, 0.6, 0.1, 0.1, 5)'$ of the data generating process \eqref{model_ptngarch}. As for the sample region $D_{NT} = \{(i,t): i = 1, 2, ..., N; t = 1, 2, ..., T\}$, let $T$ increases from 200 to 2000, while $N$ also increases at relatively slower rates of $\bigO(\sqrt{T})$ and $\bigO(T/\log(T))$ respectively, as it is showed in the following table:
\begin{table}[htbp]
\centering
\begin{tabular}{|c|cccc|}
\hline
$T$ & 200 & 500 & 1000 & 2000\\
\hline
$N \approx \sqrt{T}$ & 14  & 22  & 31  & 44\\
\hline
$N \approx T/\log(T)$ & 37  & 80  & 144  & 263\\
\hline
\end{tabular}
\end{table}
\noindent For each network size $N$, the adjacency matrix $A$ is simulated according to four different mechanisms in Example \ref{example_band} to Example \ref{example_block}.
\begin{remark}
    Particularly in the empirical analysis we will study the dataset of car collisions across different neighbourhoods that are distributed on five boroughs of New York City. These boroughs are separated by rivers (except for Brooklyn and Queens), and neighbourhoods within the same borough are more likely to share a borderline while cross-borough connections are very rare. Therefore the network constructed with New York City neighbourhoods follows the block structure in Example \ref{example_block} with $N = 20$ and $K = 5$.
\end{remark}

Based on a simulated network, the data is generated according to \eqref{model_ptngarch}, and the true parameters are estimated by the MLE \eqref{MLE_ptngarch}. To monitor the finite performance of MLE, data generation and parameter estimation are repeated for $M = 1000$ times, for each combination of sample size $(N,T)$. The $m$-th replication produces the estimates $\hat{\theta}_m = (\hat{\omega}_m, \hat{\alpha}^{(1)}_m, \hat{\alpha}^{(2)}_m, \hat{\xi}_m, \hat{\beta}_m)'$ and $\hat{r}_m$. We use the following two measurements to evaluate the performance of simulation results:
\begin{itemize}
    \item Root-mean-square error: $RMSE_k = \sqrt{M^{-1}\sum_{m=1}^M(\hat{\theta}_{km} - \theta_{k0})^2}$,
    \item Coverage probability: $CP_k = M^{-1}\sum_{m=1}^M 1_{\left\{\theta_{k0}\in CI_{km}\right\}}$,
\end{itemize}
where $\hat{\theta}_{km}$ represents $k$th component of $\hat{\theta}_m$,  
$CI_{km}$ is the 95\% confidence interval defined as 
$$CI_{km} = \left(\hat{\theta}_{km} - z_{0.975}\widehat{SE}_{km}, \hat{\theta}_{km} + z_{0.975}\widehat{SE}_{km}\right).$$ 
Here the estimated standard error $\widehat{SE}_{km}$ is the square root of $k$-th diagonal element of $(NT)^{-1}\Sighat_{NT}^{-1}$, and $z_{0.975}$ is the 0.975th quantile of standard normal distribution. To eliminate the effect of starting points, a different initial values of $\theta$ is used for each $m$. RMSEs and CPs with different sample sizes and network simulation mechanisms are reported in Tables \ref{tab:sim_results_sqrt_ptngarch} and \ref{tab:sim_results_log_ptngarch}; We also report the mean estimates of the threshold $r_0$ in the last columns of both tables.
\begin{table}[htbp]
\centering
\resizebox{\textwidth}{!}{
\begin{tabular}{|c|cc|ccccc|c|}
\hline
& $T$ & $N$ & $\omega$ & $\alpha^{(1)}$ & $\alpha^{(2)}$ & $\xi$ & $\beta$ & $\Bar{r}$\\
\hline
\multirow{4}{*}{Example \ref{example_band}}
& 200  & 14 & 0.0696 (0.94) & 0.0203 (0.94) & 0.0278 (0.93) & 0.0170 (0.95) & 0.0256 (0.93) & 5.028 \\
                           & 500  & 22 & 0.0367 (0.96) & 0.0100 (0.95) & 0.0138 (0.95) & 0.0101 (0.93) & 0.0127 (0.95) & 5     \\
                           & 1000 & 31 & 0.0238 (0.95) & 0.0058 (0.95) & 0.0081 (0.95) & 0.0062 (0.97) & 0.0074 (0.95) & 5     \\
                           & 2000 & 44 & 0.0153 (0.95) & 0.0035 (0.95) & 0.0047 (0.95) & 0.0041 (0.96) & 0.0045 (0.95) & 5     \\
                           \hline
\multirow{4}{*}{Example \ref{example_random}}
& 200  & 14 & 0.0454 (0.95) & 0.0200 (0.95) & 0.0264 (0.94) & 0.0119 (0.96) & 0.0245 (0.94) & 5.045 \\
                           & 500  & 22 & 0.0284 (0.95) & 0.0101 (0.95) & 0.0134 (0.95) & 0.0072 (0.94) & 0.0126 (0.95) & 5.002 \\
                           & 1000 & 31 & 0.0162 (0.97) & 0.0059 (0.96) & 0.0077 (0.97) & 0.0044 (0.94) & 0.0074 (0.95) & 5     \\
                           & 2000 & 44 & 0.0112 (0.96) & 0.0034 (0.96) & 0.0047 (0.95) & 0.0029 (0.94) & 0.0043 (0.96) & 5     \\
                           \hline
\multirow{4}{*}{Example \ref{example_pwl}}
& 200  & 14 & 0.0511 (0.96) & 0.0200 (0.95) & 0.0272 (0.94) & 0.0131 (0.95) & 0.0246 (0.95) & 5.034 \\
                           & 500  & 22 & 0.0349 (0.95) & 0.0102 (0.95) & 0.0135 (0.96) & 0.0084 (0.95) & 0.0127 (0.96) & 5.001 \\
                           & 1000 & 31 & 0.0146 (0.95) & 0.0060 (0.95) & 0.0079 (0.95) & 0.0038 (0.95) & 0.0077 (0.94) & 5     \\
                           & 2000 & 44 & 0.0104 (0.95) & 0.0035 (0.95) & 0.0048 (0.94) & 0.0025 (0.95) & 0.0043 (0.96) & 5     \\
                           \hline
\multirow{4}{*}{Example \ref{example_block}}
& 200  & 14 & 0.0882 (0.95) & 0.0205 (0.95) & 0.0273 (0.95) & 0.0227 (0.94) & 0.0256 (0.93) & 5.013 \\
                           & 500  & 22 & 0.0379 (0.94) & 0.0102 (0.95) & 0.0136 (0.95) & 0.0096 (0.95) & 0.0124 (0.95) & 5     \\
                           & 1000 & 31 & 0.0218 (0.95) & 0.0060 (0.95) & 0.0078 (0.95) & 0.0055 (0.95) & 0.0073 (0.96) & 5     \\
                           & 2000 & 44 & 0.0118 (0.94) & 0.0035 (0.96) & 0.0047 (0.95) & 0.0029 (0.95) & 0.0043 (0.96) & 5     \\
                           \hline
\end{tabular}
\caption{Simulation results with different network structures ($N \approx \sqrt{T}$).}
\label{tab:sim_results_sqrt_ptngarch}
}
\end{table}

\begin{table}[htbp]
\centering
\resizebox{\textwidth}{!}{
\begin{tabular}{|c|cc|ccccc|c|}
\hline
& $T$ & $N$ & $\omega$ & $\alpha^{(1)}$ & $\alpha^{(2)}$ & $\xi$ & $\beta$ & $\Bar{r}$\\
\hline
\multirow{4}{*}{Example \ref{example_band}}
& 200  & 37  & 0.0537 (0.95) & 0.0124 (0.95) & 0.0164 (0.95) & 0.0143 (0.94) & 0.0158 (0.94) & 5.002 \\
& 500  & 80  & 0.0287 (0.96) & 0.0054 (0.94) & 0.0071 (0.95) & 0.0078 (0.95) & 0.0066 (0.95) & 5     \\
& 1000 & 144 & 0.0201 (0.95) & 0.0029 (0.94) & 0.0040 (0.93) & 0.0055 (0.95) & 0.0036 (0.94) & 5     \\
& 2000 & 263 & 0.0136 (0.95) & 0.0015 (0.94) & 0.0019 (0.95) & 0.0038 (0.95) & 0.0019 (0.93) & 5     \\
\hline
\multirow{4}{*}{Example \ref{example_random}}
& 200  & 37  & 0.0347 (0.95) & 0.0121 (0.95) & 0.0170 (0.95) & 0.0089 (0.95) & 0.0161 (0.93) & 5.008 \\
& 500  & 80  & 0.0140 (0.95) & 0.0053 (0.95) & 0.0070 (0.95) & 0.0035 (0.95) & 0.0066 (0.95) & 5     \\
& 1000 & 144 & 0.0073 (0.95) & 0.0029 (0.93) & 0.0036 (0.95) & 0.0020 (0.94) & 0.0036 (0.93) & 5     \\
& 2000 & 263 & 0.0041 (0.95) & 0.0014 (0.95) & 0.0020 (0.94) & 0.0011 (0.95) & 0.0018 (0.96) & 5     \\
\hline
\multirow{4}{*}{Example \ref{example_pwl}}
& 200  & 37  & 0.0385 (0.95) & 0.0124 (0.94) & 0.0168 (0.95) & 0.0092 (0.95) & 0.0152 (0.95) & 5.003 \\
& 500  & 80  & 0.0144 (0.95) & 0.0054 (0.95) & 0.0071 (0.94) & 0.0036 (0.95) & 0.0067 (0.95) & 5     \\
& 1000 & 144 & 0.0073 (0.94) & 0.0029 (0.94) & 0.0035 (0.96) & 0.0019 (0.94) & 0.0035 (0.95) & 5     \\
& 2000 & 263 & 0.0037 (0.95) & 0.0015 (0.95) & 0.0019 (0.96) & 0.0009 (0.95) & 0.0018 (0.95) & 5     \\
\hline
\multirow{4}{*}{Example \ref{example_block}}
& 200  & 37  & 0.0498 (0.95) & 0.0120 (0.95) & 0.0165 (0.94) & 0.0129 (0.94) & 0.0148 (0.96) & 5.011 \\
& 500  & 80  & 0.0176 (0.94) & 0.0055 (0.94) & 0.0071 (0.94) & 0.0045 (0.94) & 0.0069 (0.94) & 5     \\
& 1000 & 144 & 0.0083 (0.97) & 0.0028 (0.95) & 0.0036 (0.96) & 0.0022 (0.96) & 0.0034 (0.95) & 5     \\
& 2000 & 263 & 0.0048 (0.95) & 0.0015 (0.95) & 0.0019 (0.95) & 0.0012 (0.96) & 0.0019 (0.95) & 5     \\
\hline
\end{tabular}
\caption{Simulation results with different network structures ($N \approx T/\log(T)$).}
\label{tab:sim_results_log_ptngarch}
}
\end{table}

From Tables \ref{tab:sim_results_sqrt_ptngarch} and \ref{tab:sim_results_log_ptngarch} it can be seen that the RMSEs of $\hat{\theta}_{NT}$ decrease asymptotically toward zero, and the mean of $\hat{r}_{NT}$ is equal to $r_0 = 5$ for sufficiently large sample size. These results support the consistency of MLE \eqref{MLE_ptngarch} in Theorem \ref{theorem_C_ptngarch}. The reported CPs are close to the theoretical value $0.95$. This shows that $\widehat{SE}$ provides a reliable estimation of the true standard error of $\hat{\theta}_{NT}$. Moreover, normal Q-Q plots for the estimation results, Figures \ref{fig:qq_band_ptngarch} to \ref{fig:qq_block_ptngarch}, are presented when $T = 2000, N = 44$ and $T= 2000, N = 263$ respectively, under different network structures. These Q-Q plots provide additional evidence for the asymptotic normality of $\hat{\theta}_{NT}$ in Theorem \ref{theorem_AN_ptngarch}.

\begin{figure}[htbp]
    \centering
    \subfigure[$T = 2000, N = 44$]{
    \includegraphics[width=0.8\textwidth]{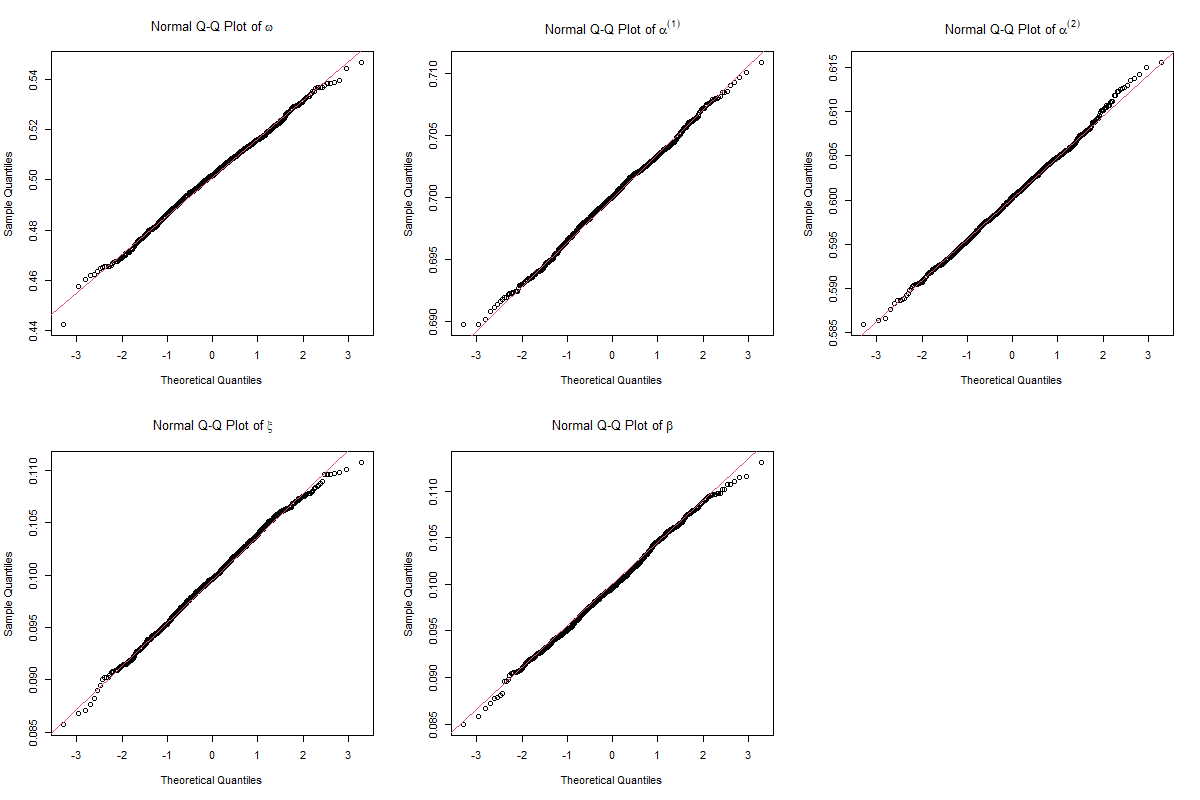}
    }
    \vfill
    \subfigure[$T = 2000, N = 263$]{
    \includegraphics[width=0.8\textwidth]{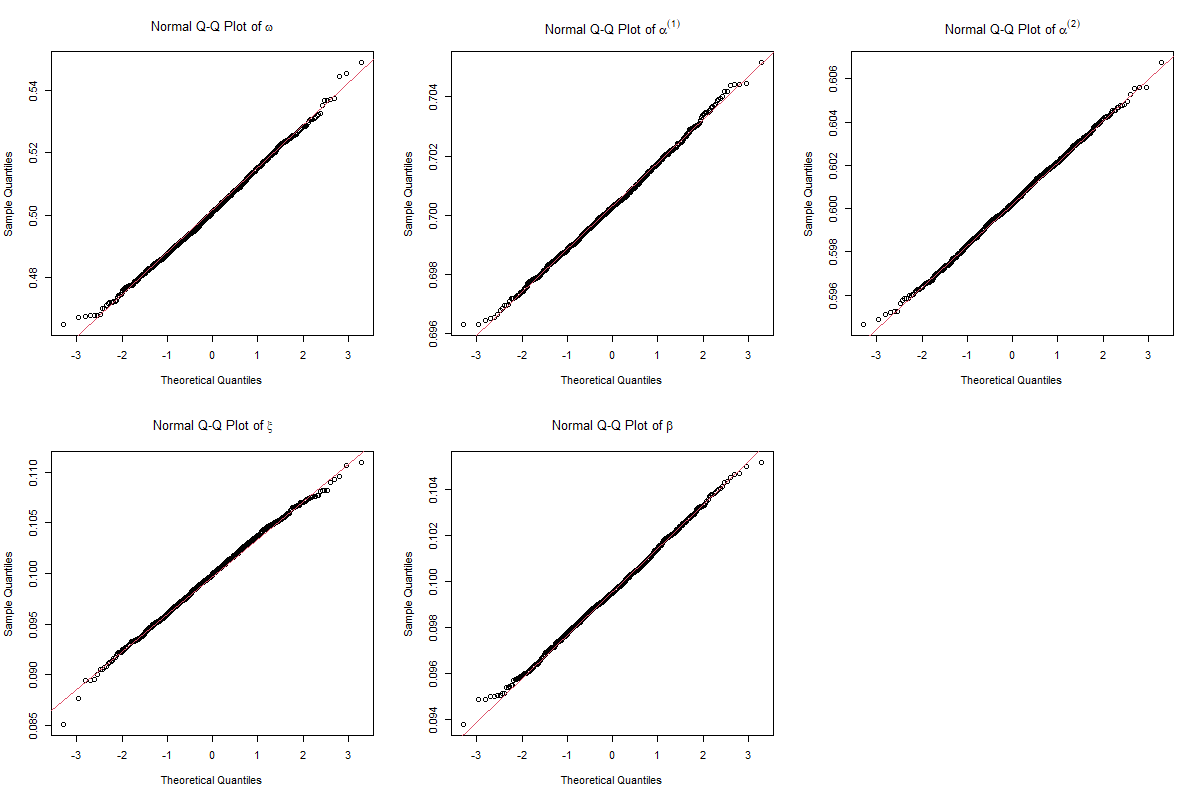}
    }
    \caption{Q-Q plots of estimates for Example \ref{example_band}.}
    \label{fig:qq_band_ptngarch}
\end{figure}

\begin{figure}[htbp]
    \centering
    \subfigure[$T = 2000, N = 44$]{
    \includegraphics[width=0.8\textwidth]{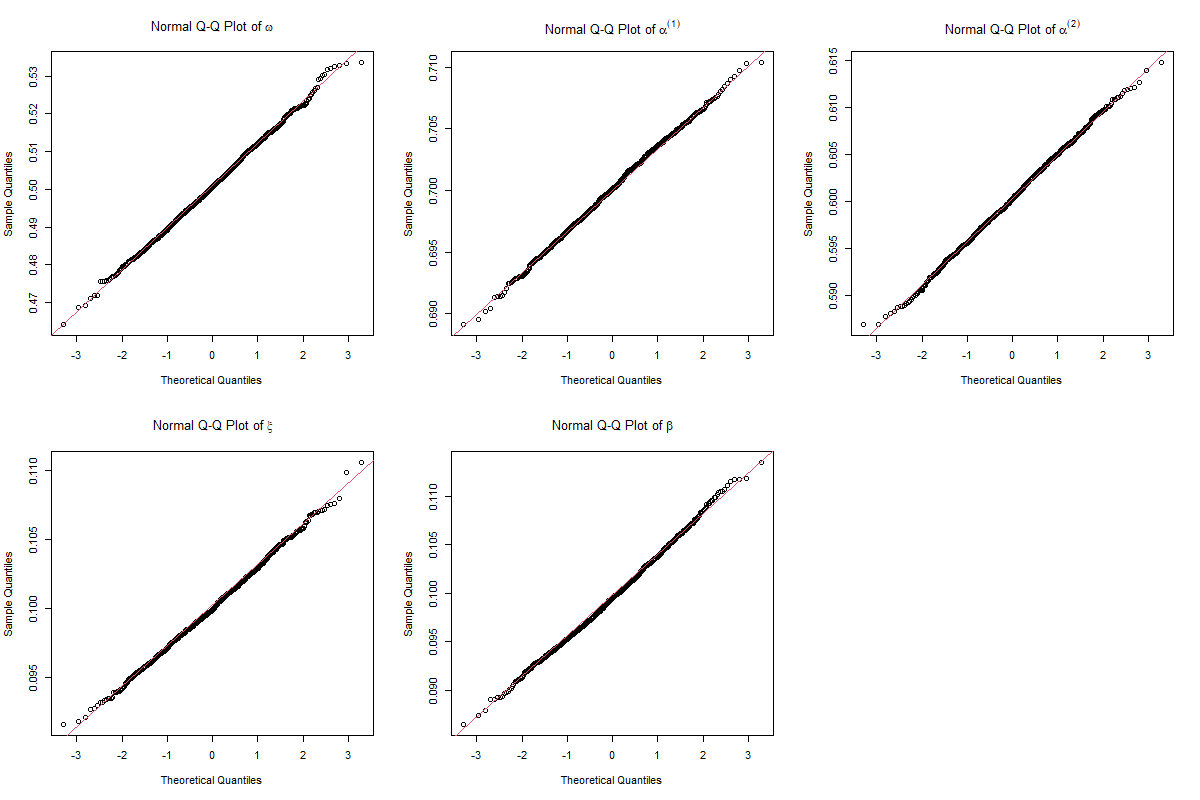}
    }
    \vfill
    \subfigure[$T = 2000, N = 263$]{
    \includegraphics[width=0.8\textwidth]{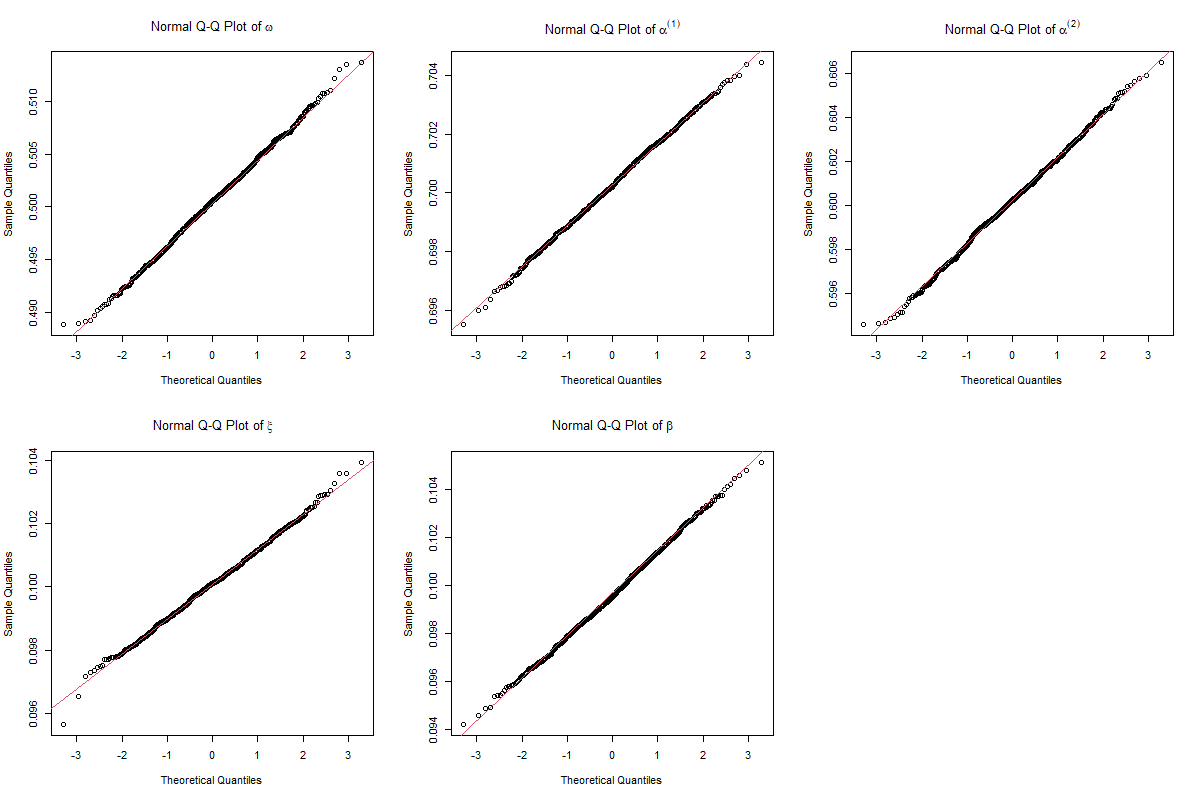}
    }
    \caption{Q-Q plots of estimates for Example \ref{example_random}.}
    \label{fig:qq_random_ptngarch}
\end{figure}

\begin{figure}[htbp]
    \centering
    \subfigure[$T = 2000, N = 44$]{
    \includegraphics[width=0.8\textwidth]{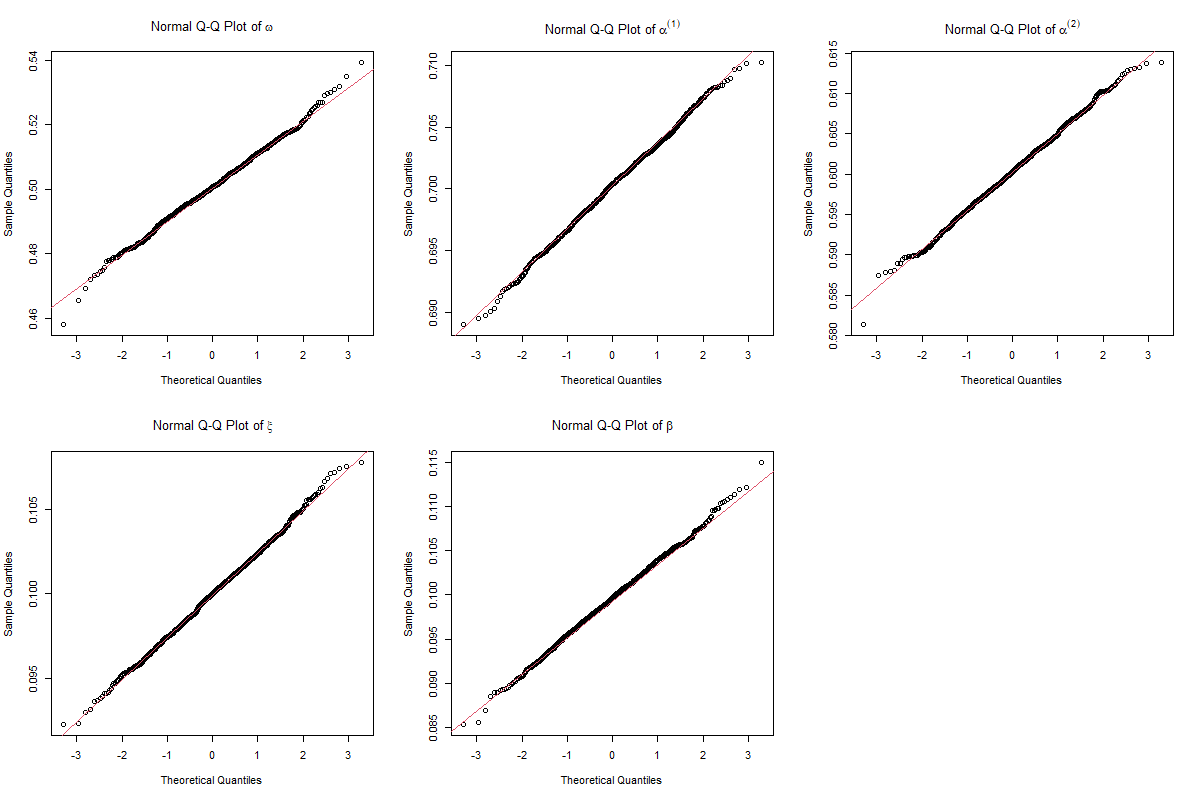}
    }
    \vfill
    \subfigure[$T = 2000, N = 263$]{
    \includegraphics[width=0.8\textwidth]{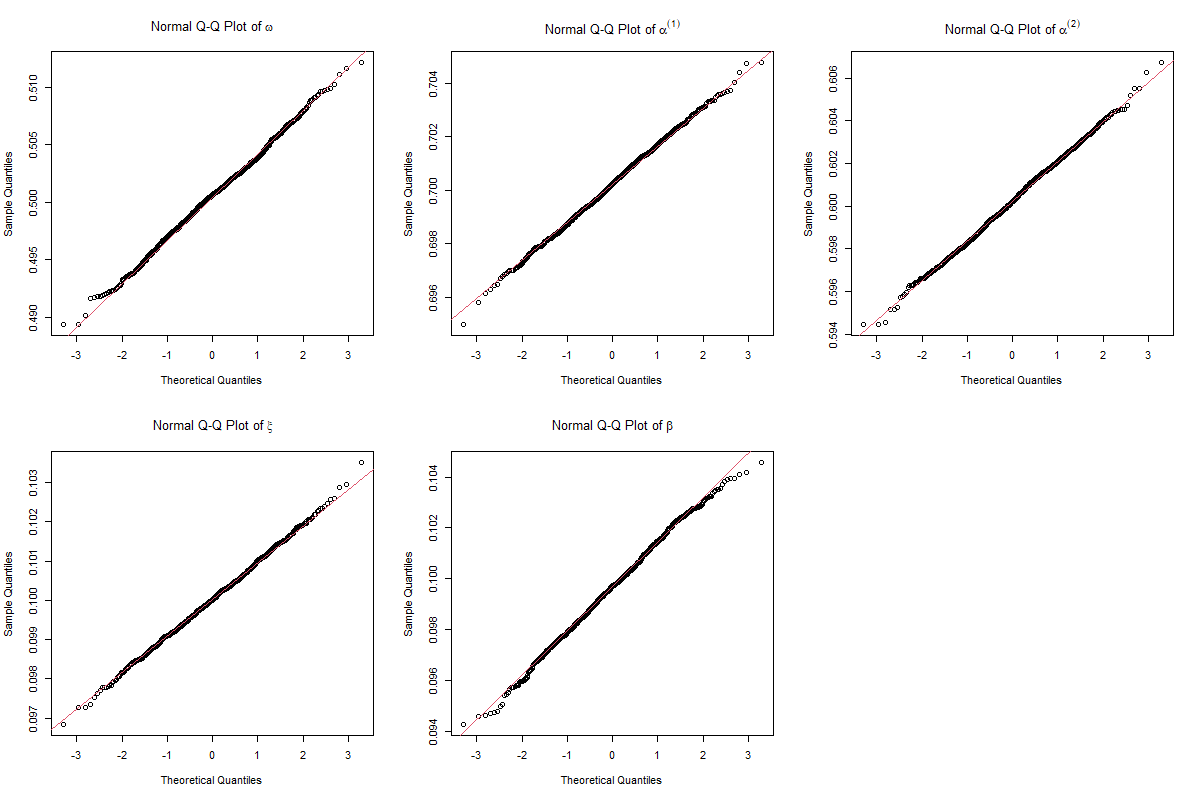}
    }
    \caption{Q-Q plots of estimates for Example \ref{example_pwl}.}
    \label{fig:qq_pwl_ptngarch}
\end{figure}

\begin{figure}[htbp]
    \centering
    \subfigure[$T = 2000, N = 44$]{
    \includegraphics[width=0.8\textwidth]{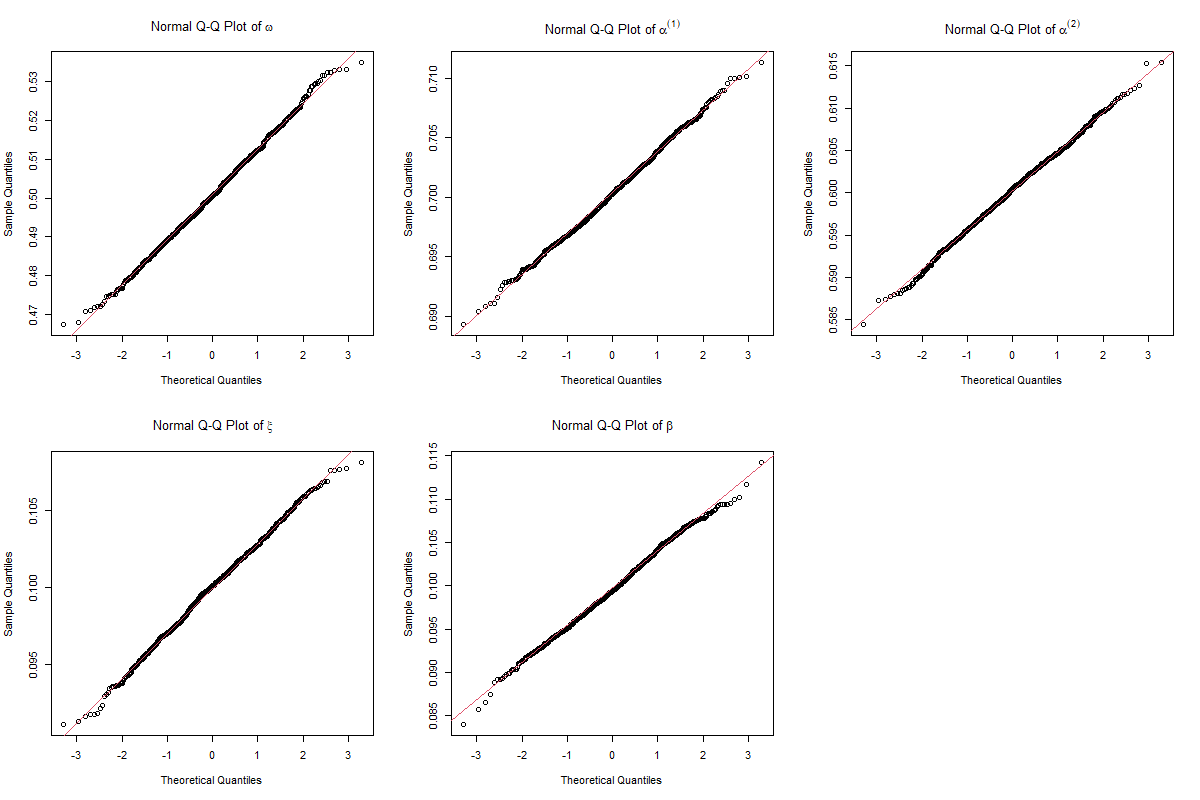}
    }
    \vfill
    \subfigure[$T = 2000, N = 263$]{
    \includegraphics[width=0.8\textwidth]{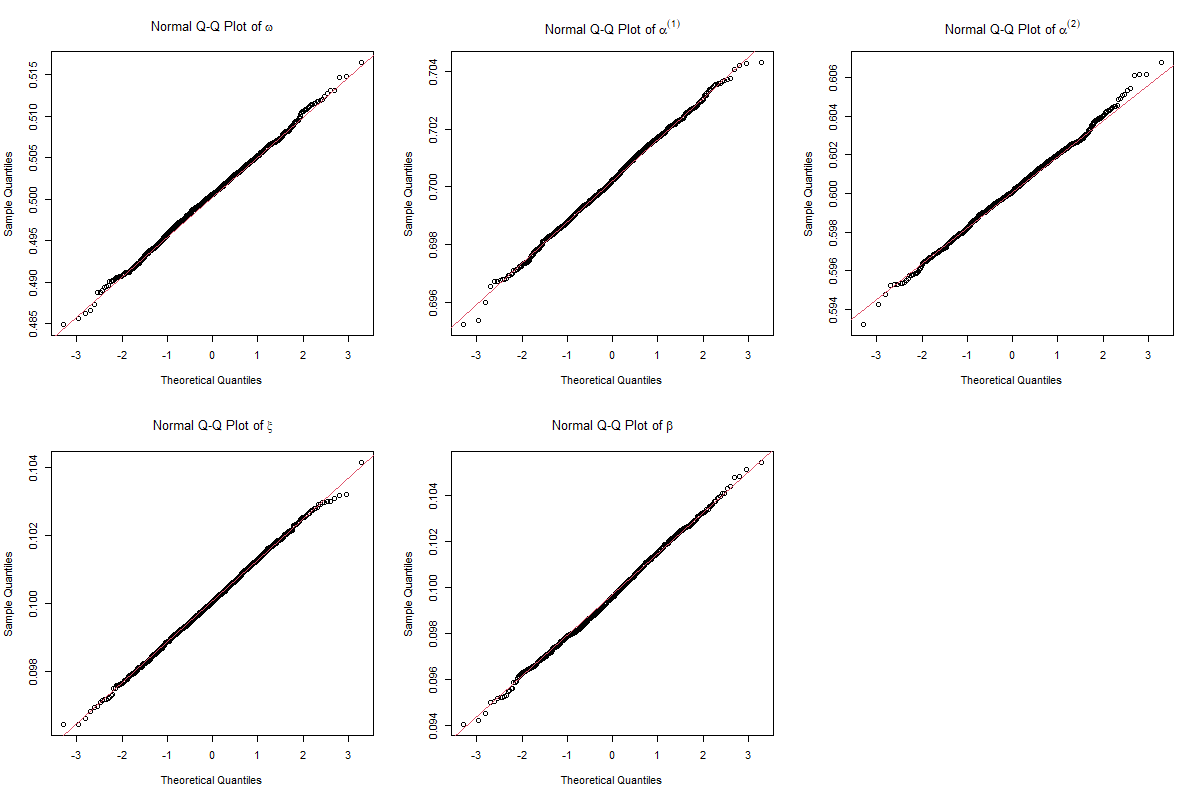}
    }
    \caption{Q-Q plots of estimates for Example \ref{example_block}.}
    \label{fig:qq_block_ptngarch}
\end{figure}

\subsection{Analysis of daily numbers of car accidents in New York City}\label{section_data_analysis_ptngarch}

New York City Police Department (NYPD) publishes and updates regularly the detailed data of motor vehicle collisions that have occurred city-wide. These data are openly accessible on NYPD website\footnote{https://www1.nyc.gov/site/nypd/stats/traffic-data/traffic-data-collision.page} and contain sufficient information for us to apply our model. We collect all records from 16th February 2021 to 30th June 2022, each record includes the date when an accident happened, and the zip code of where it happened. We classified all records into 41 neighbourhoods according to the correspondence between zip codes and the geometric locations they represent. Re-grouping the data by neighbourhoods and dates of occurrence, we obtain a high-dimensional time series with dimension $N = 41$ and sample size $T = 500$.

Two neighbourhoods are regarded as connected nodes if they share a borderline. Therefore, based on the geometric information provided by the data, we are able to construct a reasonable network with 41 nodes, which is visualized in Figure \ref{fig:network_neighbourhoods}. In Figure \ref{fig:hist_neighbourhoods} we plot histograms of daily numbers of car accidents in 9 randomly selected neighbourhoods. The shapes of histograms show potential Poisson distribution. Moreover, in Figure \ref{fig:lines_nyc} we can easily observe volatility clustering in daily numbers of car accident in four selected neighbourhoods of NYC, indicating potential autoregressive structure in the conditional heteroscedasticity of the data.

\begin{figure}[htbp]
    \centering
    \includegraphics[width=0.45\textwidth]{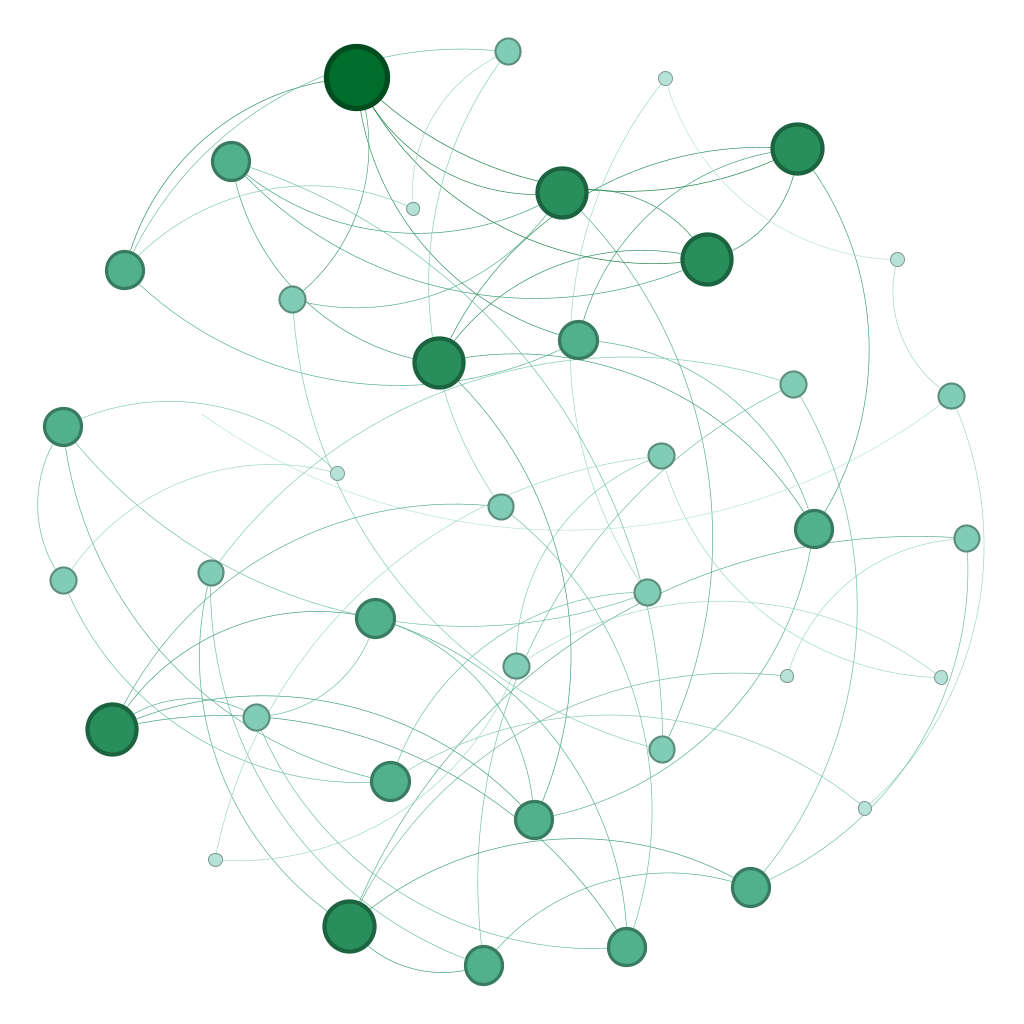}
    \caption{Network of 41 neighbourhoods in New York City}
\label{fig:network_neighbourhoods}
\end{figure}

\begin{figure}[htbp]
    \centering
    \includegraphics[width=0.8\textwidth]{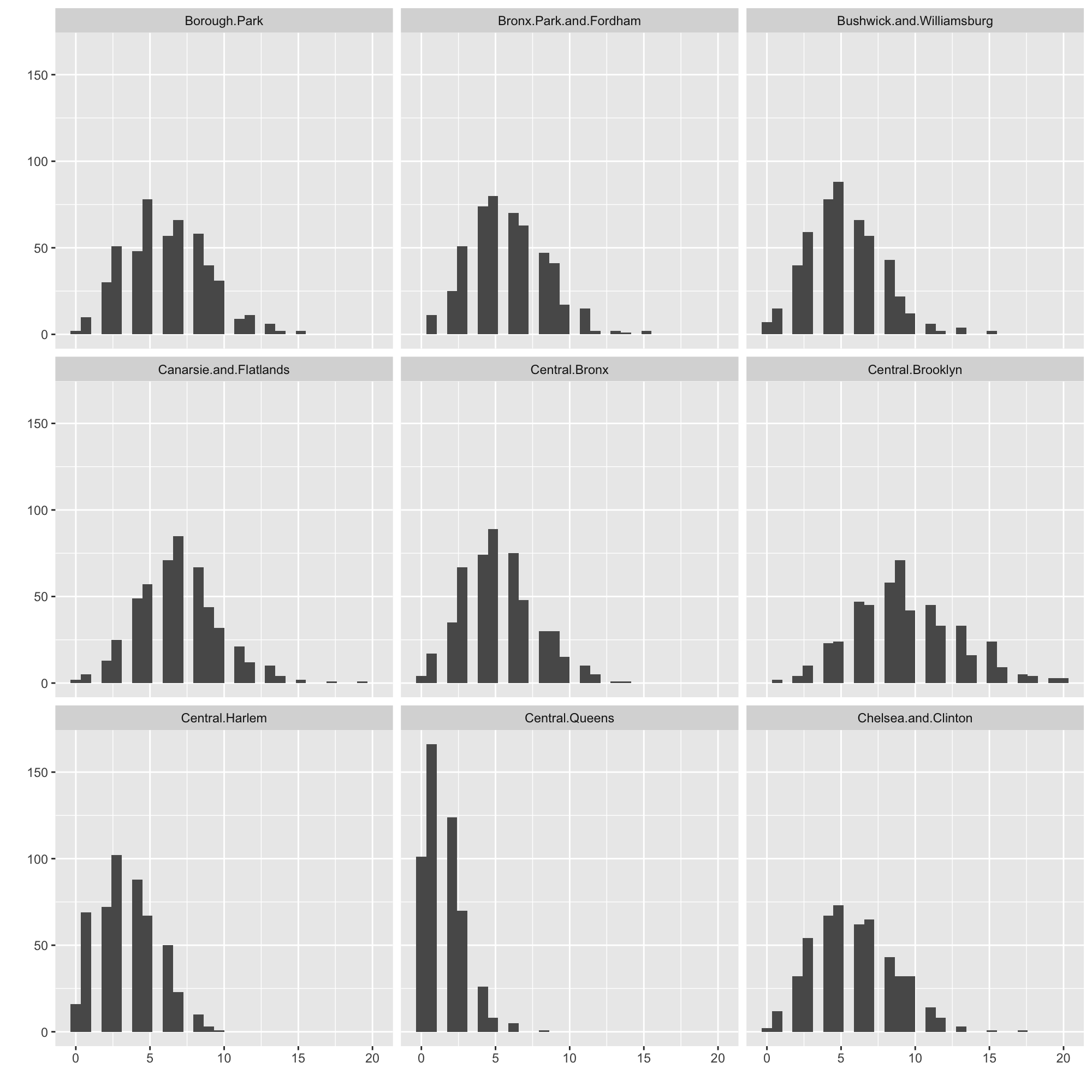}
    \caption{Distributions of daily occurrences of car accident in selected neighbourhoods.}
\label{fig:hist_neighbourhoods}
\end{figure}

\begin{figure}[htbp]
    \centering
    \includegraphics[width=\textwidth]{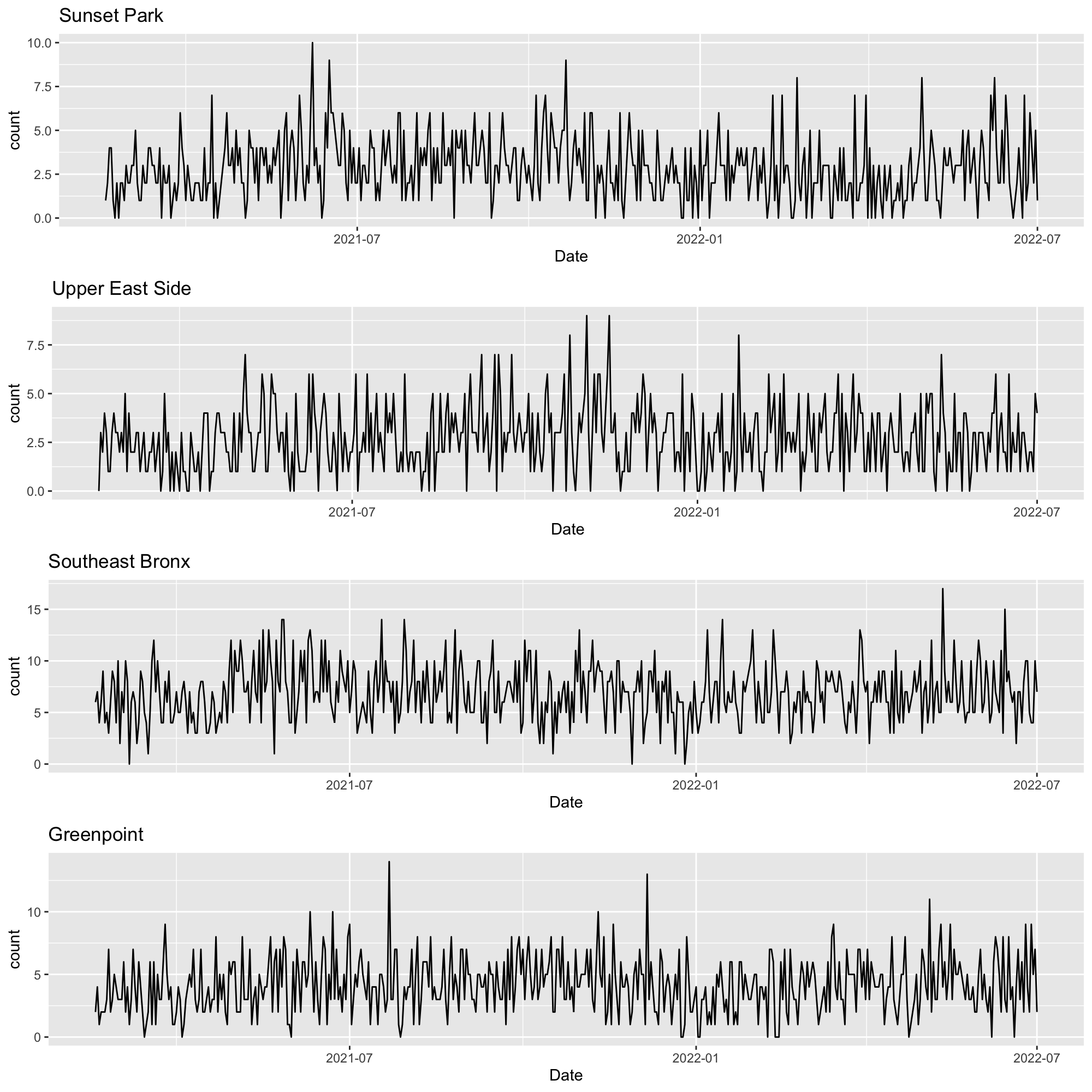}
    \caption{Daily occurrences of car accident in 4 neighbourhoods.}
    \label{fig:lines_nyc}
\end{figure}

Our model was fitted to this dataset by the method proposed in Section 3. The results of parameter estimation are reported in Table \ref{tab:real_data_results_ptngarch} below. 

\begin{table}[htbp]
\centering
\begin{tabular}{|c|c|c|c|c|c|c|}
\hline
& $\omega$ & $\alpha^{(1)}$ & $\alpha^{(2)}$ & $\xi$ & $\beta$ & $r$ \\
\hline
Estimation & 0.018693 & 0.126472 & 0.135026 & 0.002727 & 0.862244 & 10 \\
\hline
SE & 4.12e-03 & 4.40e-03 & 4.68e-03 & 1.09e-03 & 4.73e-03 & \textbackslash \\
\hline
\end{tabular}
\caption{Estimation results based on daily number of car accidents in 41 neighbourhoods of NYC.}
\label{tab:real_data_results_ptngarch}
\end{table}

Now we try to interpret these results. Firstly, it is worthy noting that $\alpha^{(1)}$ is slightly smaller than $\alpha^{(2)}$, which means that the conditional variance of the number of car accidents in these neighbourhoods are less affected by previous day's number if it is above the threshold $r = 10$. Secondly, the volatility in the number of car accidents in one area is also affected by its geometrically neighboured areas. In addition, the estimated value of $\beta$ is significantly larger than other coefficients, indicating a strong persistence in volatility that leads to volatility clustering. 

At last, we utilize the Wald test in Section 4 to investigate the following important properties of this real dataset, and these properties provide substantial evidence to support our model setting:

(i) The existence of threshold effect (i.e asymmetric property) for volatility. The null hypothesis is $H_0: \alpha^{(1)}_0 = \alpha^{(2)}_0$ (by taking $\Gamma = (0, 1, -1, 0, 0)$ and $\eta = 0$ in \eqref{null_hypo_wald_ptngarch}). In this case, the Wald statistic \eqref{wald_statistic_ptngarch} $W_{NT} = 18.94$, which suggests the rejection of $H_0$ at significant level below $0.01$ according to Theorem \ref{proposition_wald_test_ptngarch}. This testing shows that the proposed model with threshold is essentially useful for capturing the nature of daily numbers of car accidents in New York City.

(ii) The existence of GARCH effect (persistence in volatility). For the null hypothesis is $H_0: \beta_0 = 0$, the Wald statistic \eqref{wald_statistic_ptngarch} $W_{NT} = 33212$ and $p$-value is very close $0$. This strongly suggests that the autoregressive term (i.e. $\lambda_{i,t-1}$) should be included in the model. This can be seen from the large value of estimate of $\beta$. 

(iii) The existence of network effect. The Wald statistic \eqref{wald_statistic_ptngarch} $W_{NT} = 6.31$ and $p$-value is $0.012$ in this case. This indicate that there is significant network effect among the daily numbers of car accidents.

\section{Conclusion}
We have proposed a model to describe spatio-temporal integer-valued data which are observed at nodes of a network and have asymmetric property. Although the proposed model is called GARCH-type for volatility (conditional variance), it also can be used to model the conditional mean, because we assume that the conditional distribution is Poisson distribution. Asymptotic inferences, parameter estimation and hypothesis testing of the proposed model, are discussed by applying limit theorems for nonstationary arrays of random fields. Simulations for different network structures and application to a real dataset show that our methodology is useful for modelling dynamic spatio-temporal integer-valued data.              

\newpage
\begin{appendix}
\section{Proofs of theoretical results}

In this appendix, we give details of proofs for our theoretical results.

\begin{lemma}\label{lma_berkes2003_2}
    If $0 \leq \beta < 1$, $\expct\abs{y_{it}} < \infty$ and $\expct\abs{\lambda_{it}(\nu)} < \infty$ for all $(i,t)\in D_{NT}, NT\geq 1$, then
    \begin{equation}\label{lma_berkes2003_2_lambda}
    \lambda_{it}(\nu) = \sum_{k=1}^{\infty}\beta^{k-1}\left[\omega + \alpha_{i,t-k} y_{i,t-k} + \xi\sum_{j=1}^N w_{ij}y_{j,t-k}\right]
\end{equation} with probability one for all $(i,t)\in D_{NT}, NT\geq 1$ and $\nu\in\Theta\times\bZ_{+}$, where $\alpha_{i,t-k} = \alpha^{(1)}1_{\{y_{i,t-k}\geq r\}} + \alpha^{(2)}1_{\{y_{i,t-k}<r\}}$.
\end{lemma}

\begin{proof}
    When $\beta = 0$, \eqref{lma_berkes2003_2_lambda} obviously holds. Now we consider the case when $0 < \beta < 1$. Let $\log^{+}(x) = \log(x)$ if $x > 1$ and 0 otherwise, $u_{i,t-k}(\nu) := \omega + \alpha_{i,t-k} y_{i,t-k} + \xi\sum_{j=1}^N w_{ij}y_{j,t-k}$. By Jensen's inequality we have
    \begin{align*}
        &\expct\log^{+}|u_{i,t-k}(\nu)|\\
        \leq &\log^{+}\expct\left|\omega + \alpha_{i,t-k} y_{i,t-k} + \xi\sum_{j=1}^N w_{ij}y_{j,t-k}\right|\\
        < &\infty.
    \end{align*} By Lemma 2.2 in \cite{berkes2003} we have $\sum_{k=1}^{\infty}\prob\left[|u_{i,t-k}(\nu)| > \zeta^k\right] < \infty$ for any $\zeta > 1$. Therefore, by the Borel-Cantelli lemma, $|u_{i,t-k}(\nu)| \leq \zeta^k$ almost surely. Letting $1 < \zeta < \frac{1}{|\beta|}$, we can prove that the right-hand-side of \eqref{lma_berkes2003_2_lambda} converges almost surely.

    It remains for us to show that 
    $$\lambda_{it}(\nu) = \sum_{k=1}^{\infty}\beta^{k-1}u_{i,t-k}(\nu).$$ 
    From \eqref{eq_lambda_ptngarch}, we have $$\lambda_{it}(\nu) - \beta^k \lambda_{i,t-k-1}(\nu) = u_{i,t-1}(\nu) + \beta u_{i,t-2}(\nu) + ... + \beta^{k-1} u_{i,t-k}(\nu).$$ 
    Using the Markov inequality, we obtain that $\sum_{k=1}^{\infty}\prob\left\{|\beta^k \lambda_{i,t-k-1}(\nu)| > \delta\right\} < \infty$ for any $\delta > 0$. Then, by the Borel-Cantelli lemma, $|\beta^k \lambda_{i,t-k-1}(\nu)| \convas 0$ as $k\to\infty$. Letting $k\to\infty$ on both sides of above equation, we can conclude the result of Lemme \ref{lma_berkes2003_2}.
\end{proof}

\subsection{Proof of Theorem \ref{theorem_stationarity}}

Our proof of Theorem \ref{theorem_stationarity} relies on the arguments given by \cite{ferland2006} in their proof of Corollary 1. Let 
\begin{align*}
    &\Lambda_t^{(0)} := \bracketS{\lambda_{1t}^{(0)}, \lambda_{2t}^{(0)}, ..., \lambda_{Nt}^{(0)}}';\\
    &\bY_t^{(0)} := \bracketS{N_{1t}(\lambda_{1t}^{(0)}), N_{2t}(\lambda_{2t}^{(0)}), ..., N_{Nt}(\lambda_{Nt}^{(0)})}',
\end{align*} where $\{\lambda_{it}^{(0)}: i=1,2,..., N, t\in\bZ\}$ are IID positive random variables with mean 1. For each $n\geq 1$, we define $\{\bY_t^{(n)}: t\in\bZ\}$ and $\{\Lambda_t^{(n)}: t\in\bZ\}$ through following recursion:
\begin{equation}\label{proposition_stationarity_eq1}
\begin{aligned}
    &\bY_t^{(n)} = (N_{1t}(\lambda_{1t}^{(n)}), N_{2t}(\lambda_{2t}^{(n)}), ..., N_{Nt}(\lambda_{Nt}^{(n)}))';\\
    &\Lambda_t^{(n)} = \omega\vone_N + A(\bY_{t-1}^{(n-1)})\bY_{t-1}^{(n-1)} + \beta\Lambda_{t-1}^{(n-1)}.
\end{aligned}   
\end{equation}

\begin{claim}\label{clm_franke2010_lma7}
    $\{\bY_t^{(n)}: t\in\bZ\}$ is strictly stationary for each $n \geq 0$.
\end{claim}
\begin{proof}
Since $\{N_{it}(\cdot): i=1,2,..., N, t\in\bZ\}$ are independent Poisson processes with unit intensity, then for any $t$ and $h$ we have
\begin{equation}\label{clm_franke2010_lma7_eq1}
\begin{aligned}
    &\prob\bracketL{\bY_{1+h}^{(n)} = \vy_1, ..., \bY_{t+h}^{(n)} = \vy_t}\\
    = &\expct\bracketS{\prob\bracketL{\bY_{1+h}^{(n)} = \vy_1, ..., \bY_{t+h}^{(n)} = \vy_t\left|\Lambda_{1+h}^{(n)}, ..., \Lambda_{t+h}^{(n)}\right.}}\\
    = &\expct\bracketS{\prod_{k=1}^t\prod_{i=1}^N\frac{\bracketS{\lambda_{i,k+h}^{(n)}}^{y_{ik}}}{y_{ik}!}e^{-\lambda_{i,k+h}^{(n)}}}.
\end{aligned}
\end{equation}

When $n=0$, $\prob\bracketL{\bY_{1+h}^{(0)} = \vy_1, ..., \bY_{t+h}^{(0)} = \vy_t}$ is $h$-invariant for any $t$ and $h$, by \eqref{clm_franke2010_lma7_eq1} and the IID of $\{\lambda_{it}^{(0)}: i=1,2,..., N, t\in\bZ\}$. Therefore $\{\bY_t^{(0)}: t\in\bZ\}$ is strictly stationary. Assume that $\{\bY_t^{(n-1)}: t\in\bZ\}$ and $\{\Lambda_t^{(n-1)}: t\in\bZ\}$ are strictly stationary, then $\{\Lambda_t^{(n)}: t\in\bZ\}$ is also strictly stationary since $\Lambda_t^{(n)} = \omega\vone_N + A(\bY_{t-1}^{(n-1)})\bY_{t-1}^{(n-1)} + \beta\Lambda_{t-1}^{(n-1)}$. According to \eqref{clm_franke2010_lma7_eq1} and the strict stationarity of $\{\Lambda_t^{(n)}: t\in\bZ\}$, we have $\{\bY_t^{(n)}: t\in\bZ\}$ being strictly stationary too. Claim \ref{clm_franke2010_lma7} can be proved by induction.

\end{proof}

Let $\norm{\vx}_1 = |x_1| + |x_2| + ... + |x_N|$ for vector $\vx = (x_1, x_2, ..., x_N)'$. In following claim we prove the convergence of $\bY_t^{(n)}$ as $n \to \infty$.
\begin{claim}\label{clm_l1contraction_y}
    $\expct\norm{\bY_t^{(n+1)} - \bY_t^{(n)}}_1 \leq C\rho^n$ for some constants $C>0$ and $0<\rho<1$.
\end{claim}
\begin{proof}
Since $N_{it}$ is a Poisson process with unit intensity, $N_{it}(\lambda_{it}^{(n+1)}) - N_{it}(\lambda_{it}^{(n)})$ is Poisson distributed with parameter $\lambda_{it}^{(n+1)} - \lambda_{it}^{(n)}$ assuming that $\lambda_{it}^{(n+1)} \geq \lambda_{it}^{(n)}$. Then it is easy to verify that
\begin{align*}
    &\expct\norm{\bY_t^{(n+1)} - \bY_t^{(n)}}_1\\
    = &\expct\bracketM{\expct\bracketS{\norm{\bY_t^{(n+1)} - \bY_t^{(n)}}_1\left|\Lambda_t^{(n+1)},\Lambda_t^{(n)}\right.}}\\
    = &\expct\bracketM{\expct\bracketS{\sum_{i=1}^N \abs{N_{it}(\lambda_{it}^{(n+1)}) - N_{it}(\lambda_{it}^{(n)})}\left|\Lambda_t^{(n+1)},\Lambda_t^{(n)}\right.}}\\
    = &\expct\bracketM{\sum_{i=1}^N \abs{\lambda_{it}^{(n+1)} - \lambda_{it}^{(n)}}}\\
    = &\expct\norm{\Lambda_t^{(n+1)} - \Lambda_t^{(n)}}_1.
\end{align*}

Recall that, from \eqref{proposition_stationarity_eq1},  $$\Lambda_t^{(n)} = \omega\vone_N + A(\bY_{t-1}^{(n-1)})\bY_{t-1}^{(n-1)} + \beta\Lambda_{t-1}^{(n-1)}.$$ 
Then
\begin{equation}\label{clm_l1contraction_y_eq1}
\begin{aligned}
    &\norm{\bY_t^{(n+1)} - \bY_t^{(n)}}_1\\
    \leq &\norm{A(\bY_{t-1}^{(n)})\bY_{t-1}^{(n)} - A(\bY_{t-1}^{(n-1)})\bY_{t-1}^{(n-1)}}_1 + \beta\norm{\Lambda_{t-1}^{(n)} - \Lambda_{t-1}^{(n-1)}}_1.
\end{aligned}
\end{equation}
To estimate the right-side terms of the above inequality, we 
define function $\psi(y) = \alpha^{(1)} 1_{\{y \geq r\}} y + \alpha^{(2)} 1_{\{y < r\}} y$ for $y\in\bN$.  For any $y,y'\in\bN$, we have 
\begin{itemize}
    \item If $y \geq r$ and $y' \geq r$, we have $|\psi(y') - \psi(y)| = \alpha^{(1)}|y' - y| \leq \alpha^*|y' - y|$ where $\alpha^* = \max\left\{\alpha^{(1)}, \alpha^{(2)}, \abs{\alpha^{(1)}r - \alpha^{(2)}(r-1)}\right\}$;
    \item If $ y < r$ and $y' < r$, we have $|\psi(y') - \psi(y)| = \alpha^{(2)}|y' - y| \leq \alpha^*|y' - y|$.
    \item If $y$ and $y'$ are on different sides of $r$, we assume that $y \geq r$ and $y' < r$ without loss of generality. Notice that $$\frac{\psi(y) - \psi(y')}{y - y'} = \frac{\alpha^{(1)}y - \alpha^{(2)}y'}{y - y'} = \alpha^{(2)} + (\alpha^{(1)} - \alpha^{(2)})\frac{y}{y-y'}.$$ 
    When $\alpha^{(1)} \geq \alpha^{(2)}$, we see $$0 < \frac{\psi(y) - \psi(y')}{y - y'} \leq \alpha^{(2)} + (\alpha^{(1)} - \alpha^{(2)})\frac{y}{y-(r-1)} \leq \alpha^{(2)} + (\alpha^{(1)} - \alpha^{(2)})r;$$ 
    When $\alpha^{(1)} < \alpha^{(2)}$, we see 
    $$\alpha^{(2)} \geq \frac{\psi(y) - \psi(y')}{y - y'} \geq \alpha^{(2)} + (\alpha^{(1)} - \alpha^{(2)})\frac{y}{y-(r-1)} \geq \alpha^{(2)} + (\alpha^{(1)} - \alpha^{(2)})r.$$
\end{itemize}
    Combining above cases, we obtain that 
\begin{equation}\label{clm_l1contraction_y_eq_psi}
    |\psi(y') - \psi(y)| \leq \alpha^*|y' - y|
\end{equation} for any $y', y \in \bN$.

Therefore, 
\begin{align}
    &\abs{\bracketS{A(\bY_{t-1}^{(n)})\bY_{t-1}^{(n)} - A(\bY_{t-1}^{(n-1)})\bY_{t-1}^{(n-1)}}_i} \notag\\
    = &\abs{\psi(y_{i,t-1}^{(n)}) - \psi(y_{i,t-1}^{(n-1)}) + \xi\sum_{j=1}^N w_{ij}(y_{j,t-1}^{(n)} - y_{j,t-1}^{(n-1)})} \label{clm_l1contraction_y_eq2}\\
    \leq &\alpha^*\abs{y_{i,t-1}^{(n)} - y_{i,t-1}^{(n-1)}} + \xi\sum_{j=1}^N w_{ij}\abs{y_{j,t-1}^{(n)} - y_{j,t-1}^{(n-1)}} \notag
\end{align} for $i = 1, 2, ..., N$, where $(\bY)_i$ is the $i$-th element of $\bY$.

Combining \eqref{clm_l1contraction_y_eq1} and \eqref{clm_l1contraction_y_eq2}, we get
\begin{align*}
    &\expct\norm{\bY_t^{(n+1)} - \bY_t^{(n)}}_1\\
    \leq &\expct\norm{(\alpha^*I_N + \xi W + \beta I_N)(\bY_{t-1}^{(n)} - \bY_{t-1}^{(n-1)})}_1\\
    \leq &\rho(\alpha^*I_N + \xi W + \beta I_N)\expct\norm{\bY_{t-1}^{(n)} - \bY_{t-1}^{(n-1)}}_1\\
    \leq &(\alpha^* + \xi + \beta)\expct\norm{\bY_{t-1}^{(n)} - \bY_{t-1}^{(n-1)}}_1
\end{align*} 
where $\rho(\cdot)$ denotes the spectral radius, and the last inequality is due to the Gershgorin circle theorem. According to \eqref{assumption_contraction}, we can find $\rho = \alpha^* + \xi + \beta < 1$ such that
\begin{equation*}
\begin{aligned}
    &\expct\norm{\bY_t^{(n+1)} - \bY_t^{(n)}}_1\\
    \leq &\rho\expct\norm{\bY_{t-1}^{(n)} - \bY_{t-1}^{(n-1)}}_1\\
    \leq &\rho^n\expct\norm{\bY_{t-n}^{(1)} - \bY_{t-n}^{(0)}}_1\\
    = &\rho^n\expct\norm{\Lambda_{t-n}^{(1)} - \Lambda_{t-n}^{(0)}}_1\\
    \leq &C\rho^n
\end{aligned}
\end{equation*} 
for $C = \expct\norm{\Lambda_{t-n}^{(1)} - \Lambda_{t-n}^{(0)}}_1 < \infty$.

\end{proof}

By Claim \ref{clm_l1contraction_y},
\begin{equation*}
\begin{aligned}
    \prob\bracketL{\bY_t^{(n+1)} \neq \bY_t^{(n)}} = &\sum_{h=1}^{\infty}\prob\bracketL{\norm{\bY_t^{(n+1)} - \bY_t^{(n)}}_1 = h}\\
    \leq &\expct\norm{\bY_t^{(n+1)} - \bY_t^{(n)}}_1\\
    \leq &C\rho^n.
\end{aligned}
\end{equation*} 
This implies that $\sum_{n = 1}^\infty\prob\bracketL{\bY_t^{(n+1)} \neq \bY_t^{(n)}} < \infty$ and $$\prob\bracketL{\bigcap_{n=1}^{\infty}\bigcup_{k=n}^{\infty}\bracketM{\bY_t^{(k+1)} \neq \bY_t^{(k)}}} = 0$$ according to the Borel-Cantelli lemma. This means  that, with probability one,  there exists $M$ such that for all $n > M$, $\bY_t^{(n)}$ equals to some $\bY_t$ with integer components. That is, $\bY_t = \lim_{n\to\infty}\bY_t^{(n)}$ exists almost surely. Apparently, $\{\bY_t: t\in\bZ\}$ is strictly stationary since $\{\bY_t^{(n)}: t\in\bZ\}$ is strictly stationary for each $n \geq 0$, according to Claim \ref{clm_franke2010_lma7}.

At last, by Claim \ref{clm_l1contraction_y}, we also have
$$\expct\norm{\bY_t^{(n+m)} - \bY_t^{(n)}}_1 \leq \sum_{k=0}^{m-1} \expct\norm{\bY_t^{(n+k+1)} - \bY_t^{(n+k)}}_1 \leq C\rho^n\sum_{k=0}^{m-1}\rho^k,$$ for any $n, m\in\bN$. Therefore, $\{\bY_t^{(n)}: n\geq 0\}$ is a Cauchy sequence in $\bL^1$, hence $\expct\norm{\bY_t}_1 < \infty$.

\subsection{Proof of Theorem \ref{theorem_C_ptngarch}}

From Lemma \ref{lma_berkes2003_2}, we have
$$\lambda_{it}(\nu) = \sum_{k=1}^{\infty}\beta^{k-1}\left[\omega + \alpha_{i,t-k} y_{i,t-k} + \xi\sum_{j=1}^N w_{ij}y_{j,t-k}\right]$$ and
\begin{equation}\label{eq_lambda_bound}
    \sup_{NT\geq 1}\sup_{(i,t)\in D_{NT}}\sup_{\nu\in\Theta\times\bZ_{+}}|\lambda_{it}(\nu)| < \infty
\end{equation} 
with probability one, where $\alpha_{i,t-k} = \alpha^{(1)}1_{\{y_{i,t-k}\geq r\}} + \alpha^{(2)}1_{\{y_{i,t-k}<r\}}$. Given initial values $\Tilde{\lambda}_{i0} = 0$ for $i = 1, 2, ..., N$, we can replace $\lambda_{it}(\nu)$ with $\Tilde{\lambda}_{it}(\nu)$ and get $$\Tilde{\lambda}_{it}(\nu) = \sum_{k=1}^t\beta^{k-1}\left[\omega + \alpha_{i,t-k} y_{i,t-k} + \xi\sum_{j=1}^N w_{ij}y_{j,t-k}\right]$$ for $i = 1, 2, ..., N$ and $t \geq 1$. Hence
\begin{equation}\label{eq_lambda_diff}
    \lambda_{it}(\nu) - \Tilde{\lambda}_{it}(\nu) = \beta^t\lambda_{i0}(\nu).
\end{equation}

Now we are ready to prove the consistency of $\hat{\nu}_{NT}$ when $T \to \infty$ and $N \to \infty$. The proof is broken up into Claim \ref{clm_lik_func_conv_ptngarch} to Claim \ref{clm_idtfy_ptngarch} below: Claim \ref{clm_lik_func_conv_ptngarch} shows that the choice of initial values is asymptotically negligible; Claims \ref{clm_lik_func_bound_1_ptngarch} and \ref{clm_wd_lik_func_ptngarch} verify the weak dependence of $\{l_{it}(\nu): (i,t)\in D_{NT}, NT\geq 1\}$, and facilitate the adoption of LLN; Claim \ref{clm_idtfy_ptngarch} is concerned with the identifiability of the true parameters $\nu_0$.
\begin{claim}\label{clm_lik_func_conv_ptngarch}
    For any $\nu\in\Theta\times\bZ_{+}$, $|L_{NT}(\nu) - \Tilde{L}_{NT}(\nu)| \convp 0$ as $T \to \infty$ and $N \to \infty$.
\end{claim}

\begin{proof}
By \eqref{eq_lambda_bound} and \eqref{eq_lambda_diff}, we have
\begin{equation}\label{clm_lik_func_conv_eq1}
    \sup_{\nu\in\Theta\times\bZ_{+}}|\lambda_{it}(\nu) - \Tilde{\lambda}_{it}(\nu)| \leq C\rho^t
\end{equation} almost surely. Therefore,
\begin{align*}
    &|L_{NT}(\nu) - \Tilde{L}_{NT}(\nu)|\\
    \leq &\frac{1}{NT}\sum_{(i,t)\in D_{NT}}\left|y_{it}\log\left[\frac{\lambda_{it}(\nu)}{\Tilde{\lambda}_{it}(\nu)}\right] - [\lambda_{it}(\nu) - \Tilde{\lambda}_{it}(\nu)]\right|\\
    \leq &\frac{1}{NT}\sum_{(i,t)\in D_{NT}}\left[y_{it}\left|\frac{\lambda_{it}(\nu) - \Tilde{\lambda}_{it}(\nu)}{\Tilde{\lambda}_{it}(\nu)}\right| + \left|\lambda_{it}(\nu) - \Tilde{\lambda}_{it}(\nu)\right|\right]\\
    \leq &\frac{1}{NT}\sum_{(i,t)\in D_{NT}}C\rho^t\left(\frac{y_{it}}{\omega} + 1\right)
\end{align*} almost surely. By the Markov inequality, for any $\delta > 0$,
\begin{align*}
    &\prob\left\{|L_{NT}(\nu) - \Tilde{L}_{NT}(\nu)| > \delta\right\}\\
    \leq &\frac{1}{\delta NT}\sum_{i=1}^N\sum_{t=1}^TC_1\rho^t\expct\left|\frac{y_{it}}{\omega} + 1\right|\\
    \leq &\frac{1}{\delta NT}\sum_{i=1}^N\sum_{t=1}^T C_2\rho^t \quad \mbox{(because of Assumption 3.2(a))}\\
    \leq &\frac{1}{\delta T}\frac{C_2\rho}{1-\rho} \rightarrow 0.
\end{align*} as $N \to \infty$ and $T \to \infty$.
\end{proof}

In the following proofs, for a random variable $X$, we denote its $\bL^p$-norm by $\norm{X}_p = (\expct|X|^p)^{1/p}$.
\begin{claim}\label{clm_lik_func_bound_1_ptngarch}
    The functions $l_{it}(\nu)$ are uniformly $\bL^p$-bounded for some $p > 1$, i.e. $$\sup_{NT\geq 1}\sup_{(i,t)\in D_{NT}}\sup_{\nu\in\Theta\times\bZ_{+}} \norm{l_{it}(\nu)}_p < \infty.$$
\end{claim}

\begin{proof}
According to the H\"{o}lder inequality, we have
\begin{align*}
    \norm{l_{it}(\nu)}_p = &\norm{y_{it}\log\lambda_{it}(\nu) - \lambda_{it}(\nu)}_p\\
    \leq &\norm{y_{it}\log\lambda_{it}(\nu)}_p + \norm{\lambda_{it}(\nu)}_p\\
    \leq &\norm{y_{it}}_{2p}\norm{\log\lambda_{it}(\nu)}_{2p} + \norm{\lambda_{it}(\nu)}_p.
\end{align*} 
Note that
\begin{align*}
    &\sup_{\nu\in\Theta\times\bZ_{+}}\norm{\log\lambda_{it}(\nu)}_{2p}\\
    \leq &\sup_{\nu\in\Theta\times\bZ_{+}}\norm{\log^+\lambda_{it}(\nu)}_{2p} + \sup_{\nu\in\Theta\times\bZ_{+}}\norm{\log^-\lambda_{it}(\nu)}_{2p}\\
    \leq &\sup_{\nu\in\Theta\times\bZ_{+}}\norm{\lambda_{it}(\nu) + 1}_{2p} + \sup_{\nu\in\Theta\times\bZ_{+}}\max\{-\log(\omega),0\}.
\end{align*} 
Then, by Assumption \ref{assumption_y_ptngarch}(a) and \eqref{eq_lambda_bound}, we complete the proof.
\end{proof}

\begin{claim}\label{clm_wd_lik_func_ptngarch}
    For any element $\nu$ in the parameter space satisfying Assumption \ref{assumption_theta_ptngarch}, the array of random fields $\{l_{it}(\nu): (i,t)\in D_{NT}, NT\geq 1\}$ is $\eta$-weakly dependent with coefficient $\Bar{\eta}_0(s) \leq Cs^{-\mu_0}$ where $\mu_0 > 2$.
\end{claim}

\begin{proof}
For each $(i,t)\in D_{NT}$ and $h = 1, 2, ...$, define $\{y_{j\tau}^{(h)}: (j,\tau)\in D_{NT}, NT\geq 1\}$ such that $y_{j\tau}^{(h)} \neq y_{j\tau}$ if and only if $\rho((i,t), (j,\tau)) = h$. $$\lambda_{it}^{(h)}(\nu) = \sum_{k=1}^{\infty}\beta^{k-1}\left[\omega + \alpha_{i,t-k}^{(h)} y_{i,t-k}^{(h)} + \xi\sum_{j=1}^N w_{ij}y_{j,t-k}^{(h)}\right],$$ where $$\alpha_{i,t-k}^{(h)} = \alpha^{(1)}1_{\{y_{i,t-k}^{(h)} \geq r\}} + \alpha^{(2)}1_{\{y_{i,t-k}^{(h)} <r\}}.$$ 
Then, by \eqref{clm_l1contraction_y_eq_psi} and Assumption \ref{assumption_network_ptngarch}, we have
\begin{align}
    &|\lambda_{it}(\nu) - \lambda_{it}^{(h)}(\nu)| \notag\\
    \leq &\sum_{k=1}^{\infty}\beta^{k-1}|\alpha_{i,t-k} y_{i,t-k} - \alpha_{i,t-k}^{(h)} y_{i,t-k}^{(h)}| +\sum_{k=1}^{\infty}\sum_{j=1}^N\beta^{k-1}\xi w_{ij}|y_{j,t-k} - y_{j,t-k}^{(h)}| \notag\\
    = &\beta^{h-1}|\alpha_{i,t-h} y_{i,t-h} - \alpha_{i,t-h}^{(h)} y_{i,t-h}^{(h)}| + \xi\beta^{h-1}\sum_{1\leq|j-i|\leq h}w_{ij}|y_{j,t-h} - y_{j,t-h}^{(h)}| \label{clm_wd_lik_func_ptngarch_eq1}\\
    &+ \xi w_{i,i\pm h}\sum_{k=1}^h\beta^{k-1}|y_{i\pm h,t-k} - y_{i\pm h,t-k}^{(h)}| \notag\\
    \leq &\alpha^*\beta^{h-1}|y_{i,t-h} - y_{i,t-h}^{(h)}| + \xi\beta^{h-1}\sum_{1\leq|j-i|\leq h}|y_{j,t-h} - y_{j,t-h}^{(h)}| \notag\\
    &+ C\xi h^{-b}\sum_{k=1}^h|y_{i\pm h,t-k} - y_{i\pm h,t-k}^{(h)}|. \notag
\end{align} 
Therefore $\lambda_{it}(\nu)$ satisfies condition (2.7) in \cite{pan2024ltrf} with $B_{(i,t),NT}(h) \leq Ch^{-b}$ and $l = 0$. By Proposition 2 and Example 2.1 in \cite{pan2024ltrf}, the array of random fields $\{\lambda_{it}(\nu): (i,t)\in D_{NT}, NT\geq 1\}$ is $\eta$-weakly dependent with coefficient $\Bar{\eta}_{\lambda}(s) \leq Cs^{-\mu_y+2}$.

Similarly we can define $$l_{it}^{(h)}(\nu) = y_{it}^{(h)}\log\lambda_{it}^{(h)}(\nu) - \lambda_{it}^{(h)}(\nu).$$ Since
\begin{align*}
    |l_{it}(\nu) - l_{it}^{(h)}(\nu)| \leq &y_{it}\left|\log\frac{\lambda_{it}(\nu)}{\lambda_{it}^{(h)}(\nu)}\right| + |\lambda_{it}(\nu) - \lambda_{it}^{(h)}(\nu)|\\
    \leq &y_{it}\left|\frac{\lambda_{it}(\nu)}{\lambda_{it}^{(h)}(\nu)} - 1\right| + |\lambda_{it}(\nu) - \lambda_{it}^{(h)}(\nu)|\\
    \leq &\frac{y_{it}}{\omega}|\lambda_{it}(\nu) - \lambda_{it}^{(h)}(\nu)| + |\lambda_{it}(\nu) - \lambda_{it}^{(h)}(\nu)|,
\end{align*} $l_{it}(\nu)$ also satisfies condition (2.7) in \cite{pan2024ltrf} with $B_{(i,t),NT}(h) \leq Ch^{-b}$ and $l = 1$ by \eqref{clm_wd_lik_func_ptngarch_eq1}, the array of random fields $\{l_{it}(\nu): (i,t)\in D_{NT}, NT\geq 1\}$ is $\eta$-weakly dependent with coefficient $\Bar{\eta}_0(s) \leq Cs^{-\frac{2p-2}{2p-1}\mu_y+2}$. Note that $\frac{2p-2}{2p-1}\mu_y-2 > 2$ because of $\mu_y > \frac{4p-2}{p-1}$. So Claim \ref{clm_wd_lik_func_ptngarch} is proved.
\end{proof}

\begin{claim}\label{clm_idtfy_ptngarch}
    $\lambda_{it}(\nu) = \lambda_{it}(\nu_0)$ for all $(i,t)\in D_{NT}$ if and only if $\nu = \nu_0$.
\end{claim}

\begin{proof}
The \textit{if} part is obvious, it remains for us to prove the \textit{only if} part. Observe that $$(1 - \beta B)\lambda_{it}(\nu) = \omega + \alpha By_{it} + \xi\sum_{j=1}^N w_{ij}By_{jt},$$ where $B$ stands for the back-shift operator in the sense that $By_{it}^2 = y_{i,t-1}^2$, and $\alpha$ represents either $\alpha^{(1)}$ or $\alpha^{(2)}$ according to the value of $\alpha_{it}$ at time $t$. Therefore we have $$(1 - \beta B)\Lambda_t(\nu) = \omega\vone_N + (\alpha BI_N + \xi BW)\bY_t.$$

The polynomial $1 - \beta x$ has a root $x = 1/\beta$, which lies outside the unit circle since $0 < \beta < 1$. Therefore the inverse $\frac{1}{1 - \beta x}$ is well-defined for any $|x| \leq 1$, and we have $$\Lambda_t(\nu) = \frac{\omega}{1 - \beta}\vone_N + \mathcal{P}_{\nu}(B)\bY_t$$ with $\mathcal{P}_{\nu}(B) := \frac{\alpha B}{1 - \beta B}I_N + \frac{\xi B}{1 - \beta B}W.$ As $\lambda_{it}(\nu) = \lambda_{it}(\nu_0)$ for each $i = 1, 2, ..., N$, $$\left[\mathcal{P}_{\nu}(B) - \mathcal{P}_{\nu_0}(B)\right]\bY_t = \left(\frac{\omega_0}{1 - \beta_0} - \frac{\omega}{1 - \beta}\right)\vone_N.$$ We can deduce from above equation that $\mathcal{P}_{\nu}(x) = \mathcal{P}_{\nu_0}(x)$ for any $|x| \leq 1$, otherwise $\bY_t$ will be degenerated to a deterministic vector given $\cH_{t-1}$. $\mathcal{P}_{\nu}(x) = \mathcal{P}_{\nu_0}(x)$ implies that
\begin{equation*}
    \frac{\alpha x}{1 - \beta x}I_N - \frac{\alpha_0 x}{1 - \beta_0 x}I_N = \left(\frac{\xi_0 x}{1 - \beta_0 x} - \frac{\xi x}{1 - \beta x}\right)W.
\end{equation*} 
The diagonal elements of $W$ are all zeros while the matrix on the left side of above equation has non-zero diagonal elements, so we have
\begin{align*}
    &\frac{\alpha x}{1 - \beta x} = \frac{\alpha_0 x}{1 - \beta_0 x},\\
    &\frac{\xi x}{1 - \beta x} = \frac{\xi_0 x}{1 - \beta_0 x},
\end{align*} which imply $\alpha = \alpha_0$, $\beta = \beta_0$ and $\xi = \xi_0$. Besides, $\omega = \omega_0$ can be easily derived from $\frac{\omega}{1 - \beta} = \frac{\omega_0}{1 - \beta_0}$.
\end{proof}

With Claim \ref{clm_lik_func_bound_1_ptngarch} and Claim \ref{clm_wd_lik_func_ptngarch}, we can apply Theorem 1 in \cite{pan2024ltrf} and obtain that
\begin{equation}\label{lln_lik_func}
    \left[L_{NT}(\nu) - \expct L_{NT}(\nu)\right] \convp 0
\end{equation} for any $\nu\in\Theta\times\bZ_{+}$. Therefore, we have
\begin{align}
    &\lim_{T,N\to\infty}[L_{NT}(\nu) - L_{NT}(\nu_0)] \notag\\
    = &\lim_{T,N\to\infty}\expct\left[L_{NT}(\nu) - L_{NT}(\nu_0)\right] \notag\\
    = &\lim_{T,N\to\infty}\frac{1}{NT}\sum_{(i,t)\in D_{NT}}\expct\left[y_{it}\log\frac{\lambda_{it}(\nu)}{\lambda_{it}(\nu_0)} - (\lambda_{it}(\nu) - \lambda_{it}(\nu_0))\right] \notag\\
    = &\lim_{T,N\to\infty}\frac{1}{NT}\sum_{(i,t)\in D_{NT}}\expct\left\{\expct\left[\left.y_{it}\log\frac{\lambda_{it}(\nu)}{\lambda_{it}(\nu_0)} - (\lambda_{it}(\nu) - \lambda_{it}(\nu_0))\right|\lambda_{it}(\nu), \lambda_{it}(\nu_0)\right]\right\} \label{eq_idtfy}\\
    = &\lim_{T,N\to\infty}\frac{1}{NT}\sum_{(i,t)\in D_{NT}}\expct\left[\lambda_{it}(\nu_0)\log\frac{\lambda_{it}(\nu)}{\lambda_{it}(\nu_0)} - (\lambda_{it}(\nu) - \lambda_{it}(\nu_0))\right] \notag\\
    \leq &\lim_{T,N\to\infty}\frac{1}{NT}\sum_{(i,t)\in D_{NT}}\expct\left\{\lambda_{it}(\nu_0)\left[\frac{\lambda_{it}(\nu)}{\lambda_{it}(\nu_0)} - 1\right] - (\lambda_{it}(\nu) - \lambda_{it}(\nu_0))\right\} \notag\\
    = &0, \notag
\end{align} with equality if and only if $\lambda_{it}(\nu) = \lambda_{it}(\nu_0)$ for all $(i,t)\in D_{NT}$, which is equivalent to $\nu = \nu_0$ by Claim \ref{clm_idtfy_ptngarch}.

Note that Claim \ref{clm_lik_func_conv_ptngarch} implies that $$\lim_{T,N\to\infty}\prob\left[|L_{NT}(\hat{\nu}_{NT})-\Tilde{L}_{NT}(\hat{\nu}_{NT})|<\frac{\delta}{3}\right] = 1$$ for any $\delta > 0$, hence $$ \lim_{T,N\to\infty}\prob\left[L_{NT}(\hat{\nu}_{NT}) > \Tilde{L}_{NT}(\hat{\nu}_{NT}) - \frac{\delta}{3}\right] = 1.$$ Since $\hat{\nu}_{NT}$ maximizes $\Tilde{L}_{NT}(\nu)$, we have $$\lim_{T,N\to\infty}\prob\left[\Tilde{L}_{NT}(\hat{\nu}_{NT}) > \Tilde{L}_{NT}({\nu}_0) - \frac{\delta}{3}\right] = 1.$$
So $$\lim_{T,N\to\infty}\prob\left[L_{NT}(\hat{\nu}_{NT}) > \Tilde{L}_{NT}({\nu}_0) - \frac{2\delta}{3}\right] = 1.$$ Furthermore, from Claim \ref{clm_lik_func_conv_ptngarch}, $$\lim_{T,N\to\infty}\prob\left[\Tilde{L}_{NT}({\nu}_0) > L_{NT}({\nu}_0) - \frac{\delta}{3}\right] = 1.$$ Therefore we have
\begin{equation}\label{proof_proposition_MLE_consistency_eq2}
    \lim_{T,N\to\infty}\prob\left[0\leq L_{NT}(\nu_0) - L_{NT}(\hat{\nu}_{NT}) < \delta\right] = 1.
\end{equation}

Let $V_k(\theta)$ be an open sphere with centre $\theta$ and radius $1/k$. By \eqref{eq_idtfy} we could find $$\delta = \inf_{\substack{\theta\in \Theta\smallsetminus V_k(\theta_0)\\r\neq r_0}}\left[L_{NT}(\nu_0)-L_{NT}(\nu)\right] > 0.$$ Then by \eqref{proof_proposition_MLE_consistency_eq2},
\begin{align*}
    \lim_{T,N\to\infty}\prob\left\{0\leq L_{NT}(\nu_0) - L_{NT}(\hat{\nu}_{NT}) < \inf_{\substack{\theta\in \Theta\smallsetminus V_k(\theta_0)\\r\neq r_0}}\left[L_{NT}(\nu_0)-L_{NT}(\nu)\right]\right\} = 1.
\end{align*}
This implies that
\begin{align*}
    \lim_{T,N\to\infty}\prob\left[\hat{\theta}_{NT}\in V_k(\theta_0), \hat{r}_{NT} = r_0\right] = 1
\end{align*}
for any given $k>0$. Let $k\to\infty$ we complete the proof.

\subsection{Proof of Theorem \ref{theorem_AN_ptngarch}}

With a fixed threshold parameter $r = r_0$, we rewrite $\hat{\theta}_{NT} := \hat{\theta}_{NT}^{(r_0)}$, $\lambda_{it}(\theta) := \lambda_{it}(\theta, r_0)$ and $l_{it}(\theta) := l_{it}(\theta, r_0)$ etc., in succeeding proofs for simplicity of notations. To prove the asymptotic normality, we need to derive some intermediate results regarding the first, second and third order derivatives of $\lambda_{it}(\theta)$. 

Since
\begin{equation*}
    \lambda_{it}(\theta) = \sum_{k=1}^{\infty}\beta^{k-1}\left[\omega + \bracketS{\alpha^{(1)}1_{\{y_{i,t-k}\geq r\}} + \alpha^{(2)}1_{\{y_{i,t-k}<r\}}} y_{i,t-k} + \xi\sum_{j=1}^N w_{ij}y_{j,t-k}\right]
\end{equation*} almost surely, the partial derivative of $\lambda_{it}(
\theta)$ are
\begin{align}
    &\frac{\partial\lambda_{it}(\theta)}{\partial\omega} = \sum_{k=1}^{\infty}\beta^{k-1}, \notag\\
    &\frac{\partial\lambda_{it}(\theta)}{\partial\alpha^{(1)}} = \sum_{k=1}^{\infty}\beta^{k-1}y_{i,t-k}1_{\{y_{i,t-k}\geq r\}}, \notag\\
    &\frac{\partial\lambda_{it}(\theta)}{\partial\alpha^{(2)}} = \sum_{k=1}^{\infty}\beta^{k-1}y_{i,t-k}1_{\{y_{i,t-k}< r\}}, \label{eq_lambda_dr1}\\
    &\frac{\partial\lambda_{it}(\theta)}{\partial\xi} = \sum_{k=1}^{\infty}\beta^{k-1}\left(\sum_{j=1}^N w_{ij}y_{j,t-k}\right), \notag\\
    &\frac{\partial\lambda_{it}(\theta)}{\partial\beta} = \sum_{k=2}^{\infty}(k-1)\beta^{k-2}u_{i,t-k}(\theta),\notag
\end{align} where $$u_{i,t-k}(\theta) = \omega + \alpha^{(1)} y_{i,t-k}1_{\{y_{i,t-k}\geq r\}} + \alpha^{(2)} y_{i,t-k}1_{\{y_{i,t-k}< r\}} + \xi\sum_{j=1}^N w_{ij}y_{j,t-k}.$$ 
Note that
\begin{equation}\label{eq_lambda_dr1_diff}
    \frac{\partial\lambda_{it}(\theta)}{\partial\theta} - \frac{\partial\Tilde{\lambda}_{it}(\theta)}{\partial\theta} = t\beta^{t-1}\lambda_{i0}(\nu)\mathbf{e}_5 + \beta^t\frac{\partial\lambda_{i0}(\theta)}{\partial\theta},
\end{equation} where $\mathbf{e}_5 = (0, 0, 0, 0, 1)'$.

For the second order derivatives, we get that, for any $\theta_m, \theta_n \in \{\omega, \alpha^{(1)}, \alpha^{(2)}, \xi\}$, $$\frac{\partial^2\lambda_{it}(\theta)}{\partial\theta_m\partial\theta_n} = 0,$$ 
and
\begin{align}
    &\frac{\partial^2\lambda_{it}(\theta)}{\partial\omega\partial\beta} = \sum_{k=2}^{\infty}(k-1)\beta^{k-2}, \notag\\
    &\frac{\partial^2\lambda_{it}(\theta)}{\partial\alpha^{(1)}\partial\beta} = \sum_{k=2}^{\infty}(k-1)\beta^{k-2}y_{i,t-k}1_{\{y_{i,t-k}\geq r\}}, \notag\\
    &\frac{\partial^2\lambda_{it}(\theta)}{\partial\alpha^{(2)}\partial\beta} = \sum_{k=2}^{\infty}(k-1)\beta^{k-2}y_{i,t-k}1_{\{y_{i,t-k} < r\}}, \label{eq_lambda_dr2}\\
    &\frac{\partial^2\lambda_{it}(\theta)}{\partial\xi\partial\beta} = \sum_{k=2}^{\infty}(k-1)\beta^{k-2}\left(\sum_{j=1}^N w_{ij}y_{j,t-k}\right), \notag\\
    &\frac{\partial^2\lambda_{it}(\theta)}{\partial\beta^2} = \sum_{k=3}^{\infty}(k-1)(k-2)\beta^{k-3}u_{i,t-k}(\theta). \notag
\end{align} 
We also have
\begin{equation}\label{eq_lambda_dr2_diff}
    \frac{\partial^2\lambda_{it}(\nu)}{\partial\theta\partial\theta'} - \frac{\partial^2\Tilde{\lambda}_{it}(\nu)}{\partial\theta\partial\theta'} = t(t-1)\beta^{t-2}\lambda_{i0}(\nu)\ve_5\ve_5' + 2t\beta^{t-1}\frac{\partial\lambda_{i0}(\nu)}{\partial\theta}\ve_5' + \beta^t\frac{\partial^2\lambda_{i0}(\nu)}{\partial\theta\partial\theta'},
\end{equation} where $\mathbf{e}_5 = (0, 0, 0, 0, 1)'$.

For the third order derivatives of $\lambda_{it}(\theta)$,
\begin{align}
    &\frac{\partial^3\lambda_{it}(\theta)}{\partial\omega\partial\beta^2} = \sum_{k=3}^{\infty}(k-1)(k-2)\beta^{k-3}, \notag\\
    &\frac{\partial^3\lambda_{it}(\theta)}{\partial\alpha^{(1)}\partial\beta^2} = \sum_{k=3}^{\infty}(k-1)(k-2)\beta^{k-3}y_{i,t-k}1_{\{y_{i,t-k}\geq r\}}, \notag\\
    &\frac{\partial^3\lambda_{it}(\theta)}{\partial\alpha^{(2)}\partial\beta^2} = \sum_{k=3}^{\infty}(k-1)(k-2)\beta^{k-3}y_{i,t-k}1_{\{y_{i,t-k} < r\}}, \label{eq_lambda_dr3}\\
    &\frac{\partial^3\lambda_{it}(\theta)}{\partial\xi\partial\beta^2} = \sum_{k=3}^{\infty}(k-1)(k-2)\beta^{k-3}\left(\sum_{j=1}^N w_{ij}y_{j,t-k}\right), \notag\\
    &\frac{\partial^3\lambda_{it}(\theta)}{\partial\beta^3} = \sum_{k=4}^{\infty}(k-1)(k-2)(k-3)\beta^{k-4}u_{i,t-k}(\theta). \notag
\end{align}
Based on the consistency of $\hat{\theta}_{NT}$, we are now ready to prove asymptotic normality. We split the proof into Claim \ref{clm_dr1_app} to Claim \ref{clm_hessian_lln} below.

\begin{claim}\label{clm_dr1_app}
    For any $\theta_m\in \{\omega, \alpha^{(1)}, \alpha^{(2)}, \xi, \beta\}$, $\sqrt{NT}\abs{\frac{\partial\Tilde{L}_{NT}(\theta_0)}{\partial\theta_m} - \frac{\partial L_{NT}(\theta_0)}{\partial\theta_m}} \convp 0$ as $\min\left\{N, T\right\} \to \infty$ and $\frac{T}{N} \to \infty$.
\end{claim}

\begin{proof} Note that
\begin{equation}\label{dr1}
\left\{
\begin{aligned}
     &\frac{\partial L_{NT}(\theta)}{\partial\theta} = \frac{1}{NT}\sum_{(i,t)\in D_{NT}} \left(\frac{y_{it}}{\lambda_{it}(\theta)} - 1\right)\frac{\partial\lambda_{it}(\theta)}{\partial\theta},\\
     &\frac{\partial\lambda_{it}(\theta)}{\partial\theta} = \vh_{i,t-1} + \beta\frac{\partial\lambda_{i,t-1}(\theta)}{\partial\theta},
\end{aligned}\right.
\end{equation} where $$\vh_{i,t-1} := \left(1, y_{i,t-1}1_{\{y_{i,t-1} \geq r\}}, y_{i,t-1}1_{\{y_{i,t-1} < r\}}, \sum_{j=1}^N w_{ij} y_{j,t-1}, \lambda_{i,t-1}\right)',$$ 
and similarly
\begin{equation}\label{dr1_tilde}
\left\{
\begin{aligned}
     &\frac{\partial \Tilde{L}_{NT}(\theta)}{\partial\theta} = \frac{1}{NT}\sum_{(i,t)\in D_{NT}} \left(\frac{y_{it}}{\Tilde{\lambda}_{it}(\theta)} - 1\right)\frac{\partial\Tilde{\lambda}_{it}(\theta)}{\partial\theta},\\
     &\frac{\partial\Tilde{\lambda}_{it}(\theta)}{\partial\theta} = \Tilde{\vh}_{i,t-1} + \beta\frac{\partial\Tilde{\lambda}_{i,t-1}(\theta)}{\partial\theta}.
\end{aligned}\right.
\end{equation} 
Therefore, we have
\begin{align*}
    &\sqrt{NT}\left|\frac{\partial\Tilde{L}_{NT}(\theta_0)}{\partial\beta} - \frac{\partial L_{NT}(\theta_0)}{\partial\beta}\right|\\
    \leq &\frac{1}{\sqrt{NT}}\sum_{(i,t)\in D_{NT}}\left|y_{it}\left[\frac{\lambda_{it}(\theta_0) - \Tilde{\lambda}_{it}(\theta_0)}{\Tilde{\lambda}_{it}(\theta_0)\lambda_{it}(\theta_0)}\frac{\partial\Tilde{\lambda}_{it}(\theta_0)}{\partial\beta}\right.\right.\\
    &\left.\left. + \frac{1}{\lambda_{it}(\theta_0)}\left(\frac{\partial\Tilde{\lambda}_{it}(\theta_0)}{\partial\beta} - \frac{\partial\lambda_{it}(\theta_0)}{\partial\beta}\right)\right] - \left(\frac{\partial\Tilde{\lambda}_{it}(\theta_0)}{\partial\beta} - \frac{\partial\lambda_{it}(\theta_0)}{\partial\beta}\right)\right|\\
    \leq &\frac{1}{\sqrt{NT}}\sum_{(i,t)\in D_{NT}}\frac{y_{it}}{\omega_0^2}\abs{\lambda_{it}(\theta_0) - \Tilde{\lambda}_{it}(\theta_0)}\left|\frac{\partial\Tilde{\lambda}_{it}(\theta_0)}{\partial\beta}\right|\\
    &+\frac{1}{\sqrt{NT}}\sum_{(i,t)\in D_{NT}}\left(\frac{y_{it}}{\omega_0} + 1\right)\left|\frac{\partial\lambda_{it}(\theta_0)}{\partial\beta} - \frac{\partial\Tilde{\lambda}_{it}(\theta_0)}{\partial\beta}\right|.
\end{align*}
But, by Assumption \ref{assumption_y_ptngarch}(a) and \eqref{eq_lambda_diff},  we have
\begin{align}
    &\norm{\frac{1}{\sqrt{NT}}\sum_{(i,t)\in D_{NT}}\frac{y_{it}}{\omega_0^2}\abs{\lambda_{it}(\theta_0) - \Tilde{\lambda}_{it}(\theta_0)}\left|\frac{\partial\Tilde{\lambda}_{it}(\theta_0)}{\partial\beta}\right|}_1 \notag\\
    \leq &\frac{C_1}{\sqrt{NT}}\sum_{i=1}^N\sum_{t=1}^T\beta_0^t\norm{y_{it}}_1 \label{clm_dr1_app_eq1}\\
    \leq &\frac{C_2}{\sqrt{NT}}\sum_{i=1}^N\frac{\beta_0}{1-\beta_0} \rightarrow 0 \notag
\end{align} when $\min\left\{N, T\right\} \to \infty$ and $T/N \to \infty$. 
Then, in view of \eqref{eq_lambda_dr1_diff},
\begin{align}
    &\norm{\frac{1}{\sqrt{NT}}\sum_{(i,t)\in D_{NT}}\left(\frac{y_{it}}{\omega_0} + 1\right)\left|\frac{\partial\lambda_{it}(\theta_0)}{\partial\beta} - \frac{\partial\Tilde{\lambda}_{it}(\theta_0)}{\partial\beta}\right|}_1 \notag\\
    \leq &\frac{C_1}{\sqrt{NT}}\sum_{i=1}^N\sum_{t=1}^T t\beta_0^{t-1}\norm{\frac{y_{it}}{\omega_0} + 1}_1 + \frac{C_2}{\sqrt{NT}}\sum_{i=1}^N\sum_{t=1}^T \beta_0^t\norm{\frac{y_{it}}{\omega_0} + 1}_1 \label{clm_dr1_app_eq2}\\
    \leq &\frac{C_3}{\sqrt{NT}}\sum_{i=1}^N\sum_{t=1}^T t\beta_0^{t-1} + \frac{C_4}{\sqrt{NT}}\sum_{i=1}^N\sum_{t=1}^T \beta_0^t \notag\\
    \leq &\frac{C_3}{\sqrt{NT}}\sum_{i=1}^N\frac{1}{(1-\beta_0)^2} + \frac{C_4}{\sqrt{NT}}\sum_{i=1}^N\frac{\beta_0}{1-\beta_0} \rightarrow 0 \notag
\end{align} when $\min\left\{N, T\right\} \to \infty$ and $T/N \to \infty$. In light of \eqref{clm_dr1_app_eq1} and \eqref{clm_dr1_app_eq2},  we can prove that $$\sqrt{NT}\left|\frac{\partial\Tilde{L}_{NT}(\theta_0)}{\partial\beta} - \frac{\partial L_{NT}(\theta_0)}{\partial\beta}\right| \convp 0.$$ 
The proofs regarding partial derivatives w.r.t. $\omega$, $\alpha^{(1)}$, $\alpha^{(2)}$ and $\xi$ follow similar arguments and are therefore omitted here.
\end{proof}

\begin{claim}\label{clm_dr2_app}
    For any $\theta_m, \theta_n \in \{\omega, \alpha^{(1)}, \alpha^{(2)}, \xi, \beta\}$, $\sup_{|\theta-\theta_0| < \xi}\abs{\frac{\partial^2\Tilde{L}_{NT}(\theta)}{\partial\theta_m\partial\theta_n} - \frac{\partial^2 L_{NT}(\theta_0)}{\partial\theta_m\theta_n}} = \bigO_p(\xi)$ as $\min\left\{N, T\right\} \to \infty$ and $T/N \to \infty$.
\end{claim}

\begin{proof}
For any $\theta_m, \theta_n \in \{\omega, \alpha^{(1)}, \alpha^{(2)}, \xi, \beta\}$, we have
\begin{equation}\label{dr2}
\begin{aligned}
     &\frac{\partial^2 L_{NT}(\theta)}{\partial\theta_m\partial\theta_n}\\
     = &\frac{1}{NT}\sum_{i=1}^N\sum_{t=1}^T\left[\left(\frac{y_{it}}{\lambda_{it}(\theta)} - 1\right)\frac{\partial^2 \lambda_{it}(\theta)}{\partial\theta_m\partial\theta_n} - \frac{y_{it}}{\lambda_{it}^2(\theta)}\frac{\partial\lambda_{it}(\theta)}{\partial\theta_m}\frac{\partial\lambda_{it}(\theta)}{\partial\theta_n}\right],
\end{aligned}
\end{equation} and
\begin{equation}\label{dr2_tilde}
\begin{aligned}
     &\frac{\partial^2 \Tilde{L}_{NT}(\theta)}{\partial\theta_m\partial\theta_n}\\
     = &\frac{1}{NT}\sum_{i=1}^N\sum_{t=1}^T\left[\left(\frac{y_{it}}{\Tilde{\lambda}_{it}(\theta)} - 1\right)\frac{\partial^2 \Tilde{\lambda}_{it}(\theta)}{\partial\theta_m\partial\theta_n} - \frac{y_{it}}{\Tilde{\lambda}_{it}^2(\theta)}\frac{\partial\Tilde{\lambda}_{it}(\theta)}{\partial\theta_m}\frac{\partial\Tilde{\lambda}_{it}(\theta)}{\partial\theta_n}\right].
\end{aligned}
\end{equation} 
Note that
\begin{equation}\label{clm_dr2_app_eq1}
\begin{aligned}
    &\sup_{|\theta-\theta_0| < \xi}\left|\frac{\partial^2\Tilde{L}_{NT}(\theta)}{\partial\theta_m\partial\theta_n} - \frac{\partial^2 L_{NT}(\theta_0)}{\partial\theta_m\partial\theta_n}\right|\\
    \leq &\frac{1}{NT}\sum_{i=1}^N\sum_{t=1}^T\sup_{\theta\in\Theta}\left|\frac{\partial^2\Tilde{l}_{it}(\theta)}{\partial\theta_m\partial\theta_n} - \frac{\partial^2 l_{it}(\theta)}{\partial\theta_m\partial\theta_n}\right|\\
    &+ \frac{1}{NT}\sum_{i=1}^N\sum_{t=1}^T \sup_{|\theta-\theta_0| < \xi}\left|\frac{\partial^2 l_{it}(\theta)}{\partial\theta_m\partial\theta_n} - \frac{\partial^2 l_{it}(\theta_0)}{\partial\theta_m\partial\theta_n}\right|.
\end{aligned}
\end{equation} 
We will handle above two terms separately next. 

For the first term on the right-hand-side of \eqref{clm_dr2_app_eq1}, we see
\begin{align}
    &\norm{\frac{1}{NT}\sum_{i=1}^N\sum_{t=1}^T\sup_{\theta\in\Theta}\left|\frac{\partial^2\Tilde{l}_{it}(\theta)}{\partial\theta_m\partial\theta_n} - \frac{\partial^2 l_{it}(\theta)}{\partial\theta_m\partial\theta_n}\right|}_1 \notag\\
    \leq &\frac{1}{NT}\sum_{i=1}^N\sum_{t=1}^T\norm{y_{it}\sup_{\theta\in\Theta}\left(\frac{1}{\lambda_{it}} - \frac{1}{\Tilde{\lambda}_{it}}\right)\sup_{\theta\in\Theta}\frac{\partial^2\lambda_{it}}{\partial\theta_m\partial\theta_n}}_1 \notag\\
    &+ \frac{1}{NT}\sum_{i=1}^N\sum_{t=1}^T\norm{\sup_{\theta\in\Theta}\left(\frac{y_{it}}{\Tilde{\lambda}_{it}} - 1\right)\sup_{\theta\in\Theta}\left(\frac{\partial^2\lambda_{it}}{\partial\theta_m\partial\theta_n} - \frac{\partial^2\Tilde{\lambda}_{it}}{\partial\theta_m\partial\theta_n}\right)}_1 \notag\\
    &+ \frac{1}{NT}\sum_{i=1}^N\sum_{t=1}^T\norm{y_{it}\sup_{\theta\in\Theta}\left(\frac{\lambda^2_{it}}{\Tilde{\lambda}^2_{it}} - 1\right)\sup_{\theta\in\Theta}\frac{1}{\lambda_{it}^2}\frac{\partial\lambda_{it}}{\partial\theta_m}\frac{\partial\lambda_{it}}{\partial\theta_n}}_1 \label{clm_dr2_app_eq2}\\
    &+ \frac{1}{NT}\sum_{i=1}^N\sum_{t=1}^T\norm{\sup_{\theta\in\Theta}\frac{y_{it}}{\Tilde{\lambda}^2_{it}}\left[\frac{\partial\Tilde{\lambda}_{it}}{\partial\theta_m}\left(\frac{\partial\Tilde{\lambda}_{it}}{\partial\theta_n} - \frac{\partial\lambda_{it}}{\partial\theta_n}\right)\right]}_1 \notag\\
    &+ \frac{1}{NT}\sum_{i=1}^N\sum_{t=1}^T\norm{\sup_{\theta\in\Theta}\frac{y_{it}}{\Tilde{\lambda}^2_{it}}\left[\frac{\partial\lambda_{it}}{\partial\theta_n}\left(\frac{\partial\Tilde{\lambda}_{it}}{\partial\theta_m} - \frac{\partial\lambda_{it}}{\partial\theta_m}\right)\right]}_1 \notag\\
    := &T_1 + T_2 + T_3 + T_4 + T_5 \notag
\end{align} 
Analogous to the proof of \eqref{clm_dr1_app_eq1}, we can show that $T_1 \to 0$ and $T_3 \to 0$ as $\min\{N,T\} \to \infty$ and $T/N \to \infty$. Using \eqref{eq_lambda_dr2_diff}, we can also verify that $$T_2 \leq \frac{1}{NT}\sum_{i=1}^N\sum_{t=1}^T[C_1 t(t-1)\rho^{t-2} + C_2 t\rho^{t-1} + C_3 \rho^t]\norm{\sup_{\theta\in\Theta}\left(\frac{y_{it}}{\Tilde{\lambda}_{it}} - 1\right)}_1.$$ 
Then $T_2 \rightarrow 0$ as well. Similarly, using \eqref{eq_lambda_dr1_diff}, we obtain that $T_4 \rightarrow 0$ and $T_5 \rightarrow 0$. 

For the second term in the right-hand-side of \eqref{clm_dr2_app_eq1}, we notice that
a Taylor expansion of $\frac{\partial^2 l_{it}(\theta)}{\partial\theta_m\partial\theta_n}$ at $\theta_0$ yields that 
\begin{align}
    &\frac{1}{NT}\sum_{i=1}^N\sum_{t=1}^T \sup_{|\theta-\theta_0| < \xi}\left|\frac{\partial^2 l_{it}(\theta)}{\partial\theta_m\partial\theta_n} - \frac{\partial^2 l_{it}(\theta_0)}{\partial\theta_m\partial\theta_n}\right| \notag\\
    \leq &\frac{1}{NT}\sum_{i=1}^N\sum_{t=1}^T \xi\sup_{|\theta-\theta_0| < \xi}\left|\frac{\partial^3 l_{it}(\theta)}{\partial\theta_m\partial\theta_n\partial\theta_l}\right| \notag\\
    \leq &\frac{1}{NT}\sum_{i=1}^N\sum_{t=1}^T \xi\sup_{|\theta-\theta_0| < \xi}\abs{\frac{y_{it}}{\lambda_{it}} - 1}\abs{\frac{\partial^3 \lambda_{it}}{\partial\theta_m\partial\theta_n\partial\theta_l}} \notag\\
    &+ \frac{1}{NT}\sum_{i=1}^N\sum_{t=1}^T \xi\sup_{|\theta-\theta_0| < \xi}\abs{\frac{2y_{it}}{\lambda_{it}^3}}\abs{\frac{\partial\lambda_{it}}{\partial\theta_l}\frac{\partial\lambda_{it}}{\partial\theta_m}\frac{\partial\lambda_{it}}{\partial\theta_n}} \label{clm_dr2_app_eq3}\\
    &+ \frac{1}{NT}\sum_{i=1}^N\sum_{t=1}^T \xi\sup_{|\theta-\theta_0| < \xi}\abs{\frac{y_{it}}{\lambda_{it}^2}}\abs{\frac{\partial\lambda_{it}}{\partial\theta_l}\frac{\partial^2\lambda_{it}}{\partial\theta_m\partial\theta_n}} \notag\\
    &+ \frac{1}{NT}\sum_{i=1}^N\sum_{t=1}^T \xi\sup_{|\theta-\theta_0| < \xi}\abs{\frac{y_{it}}{\lambda_{it}^2}}\abs{\frac{\partial\lambda_{it}}{\partial\theta_n}\frac{\partial^2\lambda_{it}}{\partial\theta_l\partial\theta_m}} \notag\\
    &+ \frac{1}{NT}\sum_{i=1}^N\sum_{t=1}^T \xi\sup_{|\theta-\theta_0| < \xi}\abs{\frac{y_{it}}{\lambda_{it}^2}}\abs{\frac{\partial\lambda_{it}}{\partial\theta_m}\frac{\partial^2\lambda_{it}}{\partial\theta_n\partial\theta_l}} \notag\\
    &:= B_1 + B_2 + B_3 + B_4 + B_5 \notag
\end{align} for any $\theta_l, \theta_m, \theta_n \in \{\omega, \alpha^{(1)}, \alpha^{(2)}, \xi, \beta\}$. According to Assumption \ref{assumption_y_ptngarch}(a) and \eqref{eq_lambda_dr3},  we can verify that $$\expct\abs{\frac{y_{it}}{\lambda_{it}} - 1}\abs{\frac{\partial^3 \lambda_{it}}{\partial\theta_m\partial\theta_n\partial\theta_l}} < \infty,$$ hence $B_1 = \bigO_P(\xi)$. The other terms can be verified following similar lines, in light of \eqref{eq_lambda_dr1} and \eqref{eq_lambda_dr2}.

Taking \eqref{clm_dr2_app_eq2} and \eqref{clm_dr2_app_eq3} back to \eqref{clm_dr2_app_eq1}, we complete the proof.
\end{proof}

\begin{claim}\label{clm_score_clt}
    \begin{itemize}
        \item [(a)] $\sup_{NT\geq 1}\sup_{(i,t)\in D_{NT}}\norm{\frac{\partial l_{it}(\theta_0)}{\partial\theta}}_{2p} < \infty$ for some $p > 1$;
        \item [(b)] For each $\vv\in\bR^5$ such that $|\vv| = 1$, $\left\{\vv'\frac{\partial l_{it}(\theta_0)}{\partial\theta}: (i,t)\in D_{NT}, NT \geq 1\right\}$ are $\eta$-weakly dependent, with dependence coefficient $\Bar{\eta}_1(s) \leq Cs^{-\mu_1}$ where $\mu_1 > 4\vee\frac{2p-1}{p-1}$.
    \end{itemize}
\end{claim}

\begin{proof}

Recall that, from \eqref{dr1}, $$\frac{\partial l_{it}(\theta_0)}{\partial\theta} = \frac{y_{it}}{\lambda_{it}(\theta_0)}\frac{\partial\lambda_{it}(\theta_0)}{\partial\theta} - \frac{\partial\lambda_{it}(\theta_0)}{\partial\theta}.$$ 
From Assumption \ref{assumption_y_ptngarch} and Lemma \ref{lma_berkes2003_2}, it is easy to see that (a) holds. 

Now we verify (b). In the proof of Claim \ref{clm_wd_lik_func_ptngarch}, for each $(i,t)\in D_{NT}$ and $h = 1, 2, ...$, we defined $\{y_{j\tau}^{(h)}: (j,\tau)\in D_{NT}, NT\geq 1\}$ such that $y_{j\tau}^{(h)} \neq y_{j\tau}$ if and only if $\rho((i,t), (j,\tau)) = h$. At first, we verify that $\frac{\partial l_{it}(\theta_0)}{\partial\beta}$ satisfies condition (2.7) in \cite{pan2024ltrf}. Note that
\begin{equation}\label{clm_score_clt_eq1}
\begin{aligned}
    &\left|\frac{\partial l_{it}(\theta_0)}{\partial\beta} - \frac{\partial l_{it}^{(h)}(\theta_0)}{\partial\beta}\right|\\
    \leq &y_{it}\left|\frac{1}{\lambda_{it}(\theta_0)}\frac{\partial\lambda_{it}(\theta_0)}{\partial\beta} - \frac{1}{\lambda_{it}^{(h)}(\theta_0)}\frac{\partial\lambda_{it}^{(h)}(\theta_0)}{\partial\beta}\right| + \left|\frac{\partial\lambda_{it}(\theta_0)}{\partial\beta} - \frac{\partial\lambda_{it}^{(h)}(\theta_0)}{\partial\beta}\right|\\
    \leq &\left|\frac{y_{it}}{\omega_0} + 1\right|\left|\frac{\partial\lambda_{it}(\theta_0)}{\partial\beta} - \frac{\partial\lambda_{it}^{(h)}(\theta_0)}{\partial\beta}\right| + \frac{y_{it}}{\omega_0^2}\left|\frac{\partial\lambda_{it}^{(h)}(\theta_0)}{\partial\beta}\right|\left|\lambda_{it}(\theta_0) - \lambda_{it}^{(h)}(\theta_0)\right|
\end{aligned}
\end{equation} 
and
$$\frac{\partial\lambda_{it}(\theta_0)}{\partial\beta} = \sum_{k=2}^{\infty}(k-1)\beta_0^{k-2}u_{i,t-k}(\theta_0),$$ where $$u_{i,t-k}(\theta_0) = \omega_0 + \alpha^{(1)}_0 y_{i,t-k}1_{\{y_{i,t-k}\geq r_0\}} + \alpha^{(2)}_0 y_{i,t-k}1_{\{y_{i,t-k}< r_0\}} + \xi_0\sum_{j=1}^N w_{ij}y_{j,t-k}.$$ Following analogous arguments in \eqref{clm_wd_lik_func_ptngarch_eq1}, we obtain that
\begin{equation}\label{clm_score_clt_eq2}
\begin{aligned}
    \left|\frac{\partial\lambda_{it}(\theta_0)}{\partial\beta} - \frac{\partial\lambda_{it}^{(h)}(\theta_0)}{\partial\beta}\right| \leq &\alpha^*_0(h-1)\beta_0^{h-2}|y_{i,t-h} - y_{i,t-h}^{(h)}|\\
    &+ \xi_0(h-1)\beta_0^{h-2}\sum_{1\leq|i-j|\leq h}|y_{j,t-h} - y_{j,t-h}^{(h)}|\\
    &+ Ch^{-b}\sum_{k=2}^h|y_{i\pm h,t-k} - y_{i\pm h,t-k}^{(h)}|.
\end{aligned}
\end{equation}
Combining \eqref{clm_wd_lik_func_ptngarch_eq1}, \eqref{clm_score_clt_eq1} and \eqref{clm_score_clt_eq2}, we can verify that $\frac{\partial l_{it}(\theta_0)}{\partial\beta}$ satisfies condition (2.7) in \cite{pan2024ltrf} with $B_{(i,t),NT}(h) \leq Ch^{-b}$ and $l = 1$. Partial derivatives of $l_{it}(\theta_0)$ with respect to other parameters in $\theta_0$ follows similarly. Therefore $\vv'\frac{\partial l_{it}(\theta_0)}{\partial\theta}$ satisfies condition (2.7) in \cite{pan2024ltrf} with $B_{(i,t),NT}(h) \leq Ch^{-b}$ and $l = 1$ for each $\vv\in\bR^5$.

According to Proposition 2 and Example 2.1 in \cite{pan2024ltrf}, the array of random fields $\{\vv'\frac{\partial l_{it}(\theta_0)}{\partial\theta}: (i,t)\in D_{NT}, NT\geq 1\}$ is $\eta$-weakly dependent with coefficient $\Bar{\eta}_1(s) \leq Cs^{-\frac{2p-2}{2p-1}\mu_y+2}$. But $\frac{2p-2}{2p-1}\mu_y-2 > 4\vee\frac{2p-1}{p-1}$ because $\mu_y > \frac{6p-3}{p-1}\vee\frac{(4p-3)(2p-1)}{2(p-1)^2}$. So (b) is verified.
\end{proof}

\begin{claim}\label{clm_hessian_lln}
    \begin{itemize}
        \item [(a)] $\sup_{NT\geq 1}\sup_{(i,t)\in D_{NT}}\norm{\frac{\partial^2 l_{it}(\theta_0)}{\partial\theta\partial\theta'}}_p < \infty$ for some $p > 1$;
        \item [(b)] With respect to all $\theta_m, \theta_n \in \{\omega, \alpha^{(1)}, \alpha^{(2)}, \xi, \beta\}$, $\left\{\frac{\partial^2 l_{it}(\theta_0)}{\partial\theta_m\partial\theta_n}: (i,t)\in D_{NT}, NT \geq 1\right\}$ are $\eta$-weakly dependent, with dependence coefficient $\Bar{\eta}_2(s) \leq Cs^{-\mu_2}$ where $\mu_2 > 2$.
    \end{itemize}
\end{claim}

\begin{proof}

Recall that, from \eqref{dr2}, $$\frac{\partial^2 l_{it}(\theta_0)}{\partial\theta_m\partial\theta_n} = \left(\frac{y_{it}}{\lambda_{it}(\theta_0)} - 1\right)\frac{\partial^2 \lambda_{it}(\theta_0)}{\partial\theta_m\partial\theta_n} - \frac{y_{it}}{\lambda_{it}^2(\theta_0)}\frac{\partial\lambda_{it}(\theta_0)}{\partial\theta_m}\frac{\partial\lambda_{it}(\theta_0)}{\partial\theta_n}.$$ Then Claim \ref{clm_hessian_lln}(a) could be directly obtained from Assumption \ref{assumption_y_ptngarch}(a).

Following the same idea as in  previous proofs, for each $(i,t)\in D_{NT}$ and $h = 1, 2, ...$,  we define $\{y_{j\tau}^{(h)}: (j,\tau)\in D_{NT}, NT\geq 1\}$ such that $y_{j\tau}^{(h)} \neq y_{j\tau}$ if and only if $\rho((i,t), (j,\tau)) = h$. To prove (b), we verify that $\frac{\partial^2 l_{it}(\theta_0)}{\partial\theta_m\partial\theta_n}$ satisfies condition (2.7) in \cite{pan2024ltrf}. Firstly we have
\begin{align}
    &\abs{\frac{\partial^2 l_{it}(\theta_0)}{\partial\theta_m\partial\theta_n} - \frac{\partial^2 l_{it}^{(h)}(\theta_0)}{\partial\theta_m\partial\theta_n}} \notag\\
    \leq &\abs{\frac{y_{it}}{\lambda_{it}(\theta_0)} + 1}\abs{\frac{\partial^2 \lambda_{it}(\theta_0)}{\partial\theta_m\partial\theta_n} - \frac{\partial^2 \lambda_{it}^{(h)}(\theta_0)}{\partial\theta_m\partial\theta_n}} + y_{it}\abs{\frac{\partial^2 \lambda_{it}^{(h)}(\theta_0)}{\partial\theta_m\partial\theta_n}}\abs{\frac{1}{\lambda_{it}(\theta_0)} - \frac{1}{\lambda_{it}^{(h)}(\theta_0)}} \notag\\
    &+ \frac{y_{it}}{\lambda_{it}^2(\theta_0)}\abs{\frac{\partial\lambda_{it}(\theta_0)}{\partial\theta_m}\frac{\partial\lambda_{it}(\theta_0)}{\partial\theta_n} - \frac{\partial\lambda_{it}^{(h)}(\theta_0)}{\partial\theta_m}\frac{\partial\lambda_{it}^{(h)}(\theta_0)}{\partial\theta_n}} \notag\\
    &+ \abs{\frac{\partial\lambda_{it}^{(h)}(\theta_0)}{\partial\theta_m}\frac{\partial\lambda_{it}^{(h)}(\theta_0)}{\partial\theta_n}}\abs{\frac{y_{it}}{\lambda_{it}^2(\theta_0)} - \frac{y_{it}}{(\lambda_{it}^{(h)}(\theta_0))^2}} \label{clm_hessian_lln_eq1}\\
    \leq &\abs{\frac{y_{it}}{\lambda_{it}(\theta_0)} + 1}\abs{\frac{\partial^2 \lambda_{it}(\theta_0)}{\partial\theta_m\partial\theta_n} - \frac{\partial^2 \lambda_{it}^{(h)}(\theta_0)}{\partial\theta_m\partial\theta_n}} \notag\\
    &+ \frac{y_{it}}{\lambda_{it}(\theta_0)\lambda_{it}^{(h)}(\theta_0)}\abs{\frac{\partial^2 \lambda_{it}^{(h)}(\theta_0)}{\partial\theta_m\partial\theta_n}}\abs{\lambda_{it}(\theta_0) - \lambda_{it}^{(h)}(\theta_0)} \notag\\
    &+ \frac{y_{it}}{\lambda_{it}^2(\theta_0)}\abs{\frac{\partial\lambda_{it}(\theta_0)}{\partial\theta_m}}\abs{\frac{\partial\lambda_{it}(\theta_0)}{\partial\theta_n} - \frac{\partial\lambda_{it}^{(h)}(\theta_0)}{\partial\theta_n}} \notag\\
    &+ \frac{y_{it}}{\lambda_{it}^2(\theta_0)}\abs{\frac{\partial\lambda_{it}^{(h)}(\theta_0)}{\partial\theta_n}}\abs{\frac{\partial\lambda_{it}(\theta_0)}{\partial\theta_m} - \frac{\partial\lambda_{it}^{(h)}(\theta_0)}{\partial\theta_m}} \notag\\
    &+ \frac{y_{it}}{\lambda_{it}(\theta_0)\lambda_{it}^{(h)}(\theta_0)}\abs{\frac{\partial\lambda_{it}^{(h)}(\theta_0)}{\partial\theta_m}\frac{\partial\lambda_{it}^{(h)}(\theta_0)}{\partial\theta_n}}\abs{\frac{1}{\lambda_{it}(\theta_0)} + \frac{1}{\lambda_{it}^{(h)}(\theta_0)}}\abs{\lambda_{it}(\theta_0) - \lambda_{it}^{(h)}(\theta_0)} \notag\\
    \leq &\bracketS{\frac{y_{it}}{\omega_0} + 1}\abs{\frac{\partial^2 \lambda_{it}(\theta_0)}{\partial\theta_m\partial\theta_n} - \frac{\partial^2 \lambda_{it}^{(h)}(\theta_0)}{\partial\theta_m\partial\theta_n}} + C_1\frac{y_{it}}{\omega_0^2}\abs{\lambda_{it}(\theta_0) - \lambda_{it}^{(h)}(\theta_0)} \notag\\
    &+ C_2\frac{y_{it}}{\omega_0^2}\abs{\frac{\partial\lambda_{it}(\theta_0)}{\partial\theta_n} - \frac{\partial\lambda_{it}^{(h)}(\theta_0)}{\partial\theta_n}} + C_3\frac{y_{it}}{\omega_0^2}\abs{\frac{\partial\lambda_{it}(\theta_0)}{\partial\theta_m} - \frac{\partial\lambda_{it}^{(h)}(\theta_0)}{\partial\theta_m}} \notag\\
    &+ C_4\frac{y_{it}}{\omega_0^3}\abs{\lambda_{it}(\theta_0) - \lambda_{it}^{(h)}(\theta_0)}. \notag
\end{align} Taking the second order derivative with respect to $\xi$ and $\beta$ as an example, analogous to \eqref{clm_wd_lik_func_ptngarch_eq1} and \eqref{clm_score_clt_eq2}, we get
\begin{align}
    &\abs{\frac{\partial^2 \lambda_{it}(\theta_0)}{\partial\xi\partial\beta} - \frac{\partial^2 \lambda_{it}^{(h)}(\theta_0)}{\partial\xi\partial\beta}} \notag\\
    \leq &\sum_{k=2}^{\infty}(k-1)\beta^{k-2}\abs{\sum_{j=1}^Nw_{ij}y_{j,t-k} - \sum_{j=1}^Nw_{ij}y_{j,t-k}^{(h)}} \label{clm_hessian_lln_eq2}\\
    \leq & (h-1)\beta_0^{h-2}\sum_{|i-j|\leq h}|y_{j,t-h} - y_{j,t-h}^{(h)}| \notag\\
    &+ Ch^{-b}\sum_{k=2}^h|y_{i\pm h,t-k} - y_{i\pm h,t-k}^{(h)}|. \notag
\end{align} Proofs regarding second order derivatives with respect to other parameters follow similar arguments and then are omitted here. Substituting \eqref{clm_wd_lik_func_ptngarch_eq1}, \eqref{clm_score_clt_eq2} and \eqref{clm_hessian_lln_eq2} to \eqref{clm_hessian_lln_eq1}, we have that $\frac{\partial^2 l_{it}(\theta_0)}{\partial\theta_m\partial\theta_n}$ satisfies condition (2.7) in \cite{pan2024ltrf} with $B_{(i,t),NT}(h) \leq Ch^{-b}$ and $l = 1$.

According to Proposition 2 and Example 2.1 in \cite{pan2024ltrf}, the array of random fields $\{\frac{\partial^2 l_{it}(\theta_0)}{\partial\theta_m\partial\theta_n}: (i,t)\in D_{NT}, NT\geq 1\}$ is $\eta$-weakly dependent with coefficient $\Bar{\eta}_1(s) \leq Cs^{-\frac{2p-2}{2p-1}\mu_y+2}$, and $\frac{2p-2}{2p-1}\mu_y-2 > 2$.

\end{proof}

By the Taylor expansion, for some $\theta^*$ between $\hat{\theta}_{NT}$ and $\theta_0$ we have $$\frac{\partial\Tilde{L}_{NT}(\hat{\theta}_{NT})}{\partial\theta} = \frac{\partial\Tilde{L}_{NT}(\theta_0)}{\partial\theta} + \frac{\partial^2\Tilde{L}_{NT}(\theta^*)}{\partial\theta\partial\theta'}(\hat{\theta}_{NT} - \theta_0).$$ Since $\frac{\partial\Tilde{L}_{NT}(\hat{\theta}_{NT})}{\partial\theta} = 0$, we have
\begin{align}
    &\sqrt{NT}\Sigma_{NT}^{1/2}(\hat{\theta}_{NT} - \theta_0) \notag\\
    = &-\Sigma_{NT}^{1/2}\left(\frac{\partial^2\Tilde{L}_{NT}(\theta^*)}{\partial\theta\partial\theta'}\right)^{-1}\sqrt{NT}\frac{\partial\Tilde{L}_{NT}(\theta_0)}{\partial\theta} \label{proposition_AN_ptngarch_eq1}\\
    = &-\Sigma_{NT}^{1/2}\left(\Sigma_{NT}^{-1/2}\frac{\partial^2 L_{NT}(\theta_0)}{\partial\theta\partial\theta'}\right)^{-1}\Sigma_{NT}^{-1/2}\sqrt{NT}\frac{\partial L_{NT}(\theta_0)}{\partial\theta} + o_p(1) \notag
\end{align} according to Claims \ref{clm_dr1_app} and \ref{clm_dr2_app}.

Note that $y_{it} = N_{it}(\lambda_{it}(\theta_0))$ is Poisson distributed with mean $\lambda_{it}(\theta_0)$ conditioning on historical information $\cH_{t-1}$, with $\{N_{it}: (i,t)\in D_{NT}, NT\geq 1\}$ being IID Poisson point processes with intensity 1. Therefore we have
\begin{align*}
    &\expct\left(\frac{\partial^2 L_{NT}(\theta_0)}{\partial\theta\partial\theta'}\right)\\
    = &\frac{1}{NT}\sum_{i=1}^N\sum_{t=1}^T\expct\left\{\expct\left[\left(\frac{N_{it}(\lambda_{it}(\theta_0))}{\lambda_{it}(\theta_0)} - 1\right)\frac{\partial^2 \lambda_{it}(\theta_0)}{\partial\theta\partial\theta'}|\cH_{t-1}\right]\right\}\\
    &- \frac{1}{NT}\sum_{i=1}^N\sum_{t=1}^T\expct\left\{\expct\left[\frac{N_{it}(\lambda_{it}(\theta_0))}{\lambda_{it}^2(\theta_0)}\frac{\partial\lambda_{it}(\theta_0)}{\partial\theta}\frac{\partial\lambda_{it}(\theta_0)}{\partial\theta'}|\cH_{t-1}\right]\right\}\\
    = &-\frac{1}{NT}\sum_{i=1}^N\sum_{t=1}^T\expct\left[\frac{1}{\lambda_{it}(\theta_0)}\frac{\partial\lambda_{it}(\theta_0)}{\partial\theta}\frac{\partial\lambda_{it}(\theta_0)}{\partial\theta'}\right]\\
    = &-\Sigma_{NT}.
\end{align*} 
By Claim \ref{clm_hessian_lln}, we apply Theorem 1 in \cite{pan2024ltrf} and obtain that
\begin{equation}\label{proposition_AN_ptngarch_eq_lln}
    \frac{\partial^2 L_{NT}(\theta_0)}{\partial\theta\partial\theta'} + \Sigma_{NT} \convp 0.
\end{equation} According to condition \eqref{condition_sigma_ptngarch} we can further prove that
\begin{equation}\label{proposition_AN_ptngarch_eq2}
    -\left(\Sigma_{NT}^{-1/2}\frac{\partial^2 L_{NT}(\theta_0)}{\partial\theta\partial\theta'}\right)\Sigma_{NT}^{-1/2} = \left(\Sigma_{NT}^{1/2} + o_p(1)\right)\Sigma_{NT}^{-1/2} = I_5+o_p(1).
\end{equation}
When $\tau \neq t$ or $j \neq i$, we have $$\expct\left[\left(\frac{N_{it}(\lambda_{it}(\theta_0))}{\lambda_{it}(\theta_0)} - 1\right)\left(\frac{N_{j\tau}(\lambda_{j\tau}(\theta_0))}{\lambda_{j\tau}(\theta_0)} - 1\right)\frac{\partial\lambda_{it}(\theta_0)}{\partial\theta}\frac{\partial\lambda_{j\tau}(\theta_0)}{\partial\theta'}|\cH_{t-1}\right] = 0$$ assuming $\tau < t$. Then we can verify that
\begin{align*}
    &\var\left(\sqrt{NT}\frac{\partial L_{NT}(\theta_0)}{\partial\theta}\right)\\
    = &\frac{1}{NT}\expct\left\{\left[\sum_{i=1}^N\sum_{t=1}^T\left(\frac{N_{it}(\lambda_{it}(\theta_0))}{\lambda_{it}(\theta_0)} - 1\right)\frac{\partial \lambda_{it}(\theta_0)}{\partial\theta}\right]\right.\\
    &\times\left.\left[\sum_{i=1}^N\sum_{t=1}^T\left(\frac{N_{it}(\lambda_{it}(\theta_0))}{\lambda_{it}(\theta_0)} - 1\right)\frac{\partial \lambda_{it}(\theta_0)}{\partial\theta'}\right]\right\}\\
    = &\frac{1}{NT}\sum_{i=1}^N\sum_{t=1}^T\expct\left[\left(\frac{N_{it}(\lambda_{it}(\theta_0))}{\lambda_{it}(\theta_0)} - 1\right)^2\frac{\partial\lambda_{it}(\theta_0)}{\partial\theta}\frac{\partial\lambda_{it}(\theta_0)}{\partial\theta'}\right]\\
    = &\Sigma_{NT}.
\end{align*}
For each $\vv\in\bR^5$, $\var\bracketS{\sum_{(i,t)\in D_{NT}}\vv'\frac{\partial l_{it}(\theta_0)}{\partial\theta}} = (NT)\vv'\Sigma_{NT}\vv.$ By \eqref{condition_sigma_ptngarch} and the symmetry of $\Sigma_{NT}$, it is implied that 
$$\inf_{NT\geq 1}\vv'\Sigma_{NT}\vv  > 0.$$ Then, by Claim \ref{clm_score_clt} and Theorem 2 in \cite{pan2024ltrf}, we can prove that $$[(NT)\vv'\Sigma_{NT}\vv]^{-1/2}\vv'(NT)\frac{\partial L_{NT}(\theta_0)}{\partial\theta} \convd N(0,1).$$ According to the Cram\'{e}r-Wold theorem, we have
\begin{equation}\label{proposition_AN_ptngarch_eq_clt}
    (\Sigma_{NT})^{-1/2}\sqrt{NT}\frac{\partial L_{NT}(\theta_0)}{\partial\theta} \convd N(0, I_5).
\end{equation}
Combining \eqref{proposition_AN_ptngarch_eq1}, \eqref{proposition_AN_ptngarch_eq2} and \eqref{proposition_AN_ptngarch_eq_clt}, we complete the proof of Theorem \ref{theorem_AN_ptngarch}.

\subsection{Proof of Theorem \ref{proposition_wald_test_ptngarch}}

Recalling from \eqref{wald_statistic_ptngarch}, the Wald statistic is
$$W_{NT} := (\Gamma\hat{\theta}_{NT} - \eta)'\left\{\frac{\Gamma}{NT}\Sighat_{NT}^{-1}\Gamma'\right\}^{-1}(\Gamma\hat{\theta}_{NT} - \eta),$$ where $$\Sighat_{NT} := \frac{1}{NT}\sum_{(i,t)\in D_{NT}}\left[\frac{1}{\Tilde{\lambda}_{it}(\hat{\theta}_{NT})}\frac{\partial\Tilde{\lambda}_{it}(\hat{\theta}_{NT})}{\partial\theta}\frac{\partial\Tilde{\lambda}_{it}(\hat{\theta}_{NT})}{\partial\theta'}\right].$$ It is sufficient to show that
\begin{equation}\label{proposition_wald_test_ptngarch_eq1}
    \frac{1}{NT}\sum_{(i,t)\in D_{NT}}\left[\frac{1}{\Tilde{\lambda}_{it}(\hat{\theta}_{NT})}\frac{\partial\Tilde{\lambda}_{it}(\hat{\theta}_{NT})}{\partial\theta}\frac{\partial\Tilde{\lambda}_{it}(\hat{\theta}_{NT})}{\partial\theta'}\right]  -\Sigma_{NT} \convp 0.
\end{equation}
Note that
\begin{align*}
    &\frac{1}{NT}\sum_{(i,t)\in D_{NT}}\left[\frac{1}{\Tilde{\lambda}_{it}(\hat{\theta}_{NT})}\frac{\partial\Tilde{\lambda}_{it}(\hat{\theta}_{NT})}{\partial\theta}\frac{\partial\Tilde{\lambda}_{it}(\hat{\theta}_{NT})}{\partial\theta'}\right] - \Sigma_{NT}\\
    = &\frac{1}{NT}\sum_{(i,t)\in D_{NT}}\left[\frac{1}{\Tilde{\lambda}_{it}(\hat{\theta}_{NT})}\frac{\partial\Tilde{\lambda}_{it}(\hat{\theta}_{NT})}{\partial\theta}\frac{\partial\Tilde{\lambda}_{it}(\hat{\theta}_{NT})}{\partial\theta'} - \expct\left(\frac{1}{\lambda_{it}(\theta_0)}\frac{\partial\lambda_{it}(\theta_0)}{\partial\theta}\frac{\partial\lambda_{it}(\theta_0)}{\partial\theta'}\right)\right]\\
    = &\frac{1}{NT}\sum_{(i,t)\in D_{NT}}\left[\frac{1}{\Tilde{\lambda}_{it}(\hat{\theta}_{NT})}\frac{\partial\Tilde{\lambda}_{it}(\hat{\theta}_{NT})}{\partial\theta}\frac{\partial\Tilde{\lambda}_{it}(\hat{\theta}_{NT})}{\partial\theta'} - \frac{1}{\lambda_{it}(\theta_0)}\frac{\partial\lambda_{it}(\theta_0)}{\partial\theta}\frac{\partial\lambda_{it}(\theta_0)}{\partial\theta'}\right]\\
    &+\frac{1}{NT}\sum_{(i,t)\in D_{NT}}\left[\frac{1}{\lambda_{it}(\theta_0)}\frac{\partial\lambda_{it}(\theta_0)}{\partial\theta}\frac{\partial\lambda_{it}(\theta_0)}{\partial\theta'} - \expct\left(\frac{1}{\lambda_{it}(\theta_0)}\frac{\partial\lambda_{it}(\theta_0)}{\partial\theta}\frac{\partial\lambda_{it}(\theta_0)}{\partial\theta'}\right)\right]\\
    := &T_1 + T_2.
\end{align*} Similar to the proof of Claim \ref{clm_hessian_lln}, we can verify that the LLN Theorem 1 in \cite{pan2024ltrf} applies to $\left\{\frac{1}{\lambda_{it}(\theta_0)}\frac{\partial\lambda_{it}(\theta_0)}{\partial\theta}\frac{\partial\lambda_{it}(\theta_0)}{\partial\theta'}: (i,t)\in D_{NT}, NT\geq 1\right\}$ and therefore $T_2 \convp 0$.
$T_1$ can be further decomposed as follows
\begin{align*}
    &\frac{1}{NT}\sum_{(i,t)\in D_{NT}}\left[\frac{1}{\Tilde{\lambda}_{it}(\hat{\theta}_{NT})}\frac{\partial\Tilde{\lambda}_{it}(\hat{\theta}_{NT})}{\partial\theta}\frac{\partial\Tilde{\lambda}_{it}(\hat{\theta}_{NT})}{\partial\theta'} - \frac{1}{\lambda_{it}(\theta_0)}\frac{\partial\lambda_{it}(\theta_0)}{\partial\theta}\frac{\partial\lambda_{it}(\theta_0)}{\partial\theta'}\right]\\
    = &\frac{1}{NT}\sum_{(i,t)\in D_{NT}}\left[\frac{1}{\Tilde{\lambda}_{it}(\hat{\theta}_{NT})}\frac{\partial\Tilde{\lambda}_{it}(\hat{\theta}_{NT})}{\partial\theta}\frac{\partial\Tilde{\lambda}_{it}(\hat{\theta}_{NT})}{\partial\theta'} - \frac{1}{\lambda_{it}(\hat{\theta}_{NT})}\frac{\partial\lambda_{it}(\hat{\theta}_{NT})}{\partial\theta}\frac{\partial\lambda_{it}(\hat{\theta}_{NT})}{\partial\theta'}\right]\\
    &+ \frac{1}{NT}\sum_{(i,t)\in D_{NT}}\left[\frac{1}{\lambda_{it}(\hat{\theta}_{NT})}\frac{\partial\lambda_{it}(\hat{\theta}_{NT})}{\partial\theta}\frac{\partial\lambda_{it}(\hat{\theta}_{NT})}{\partial\theta'} - \frac{1}{\lambda_{it}(\theta_0)}\frac{\partial\lambda_{it}(\theta_0)}{\partial\theta}\frac{\partial\lambda_{it}(\theta_0)}{\partial\theta'}\right]\\
    := &S_1 + S_2.
\end{align*} But, $S_2 \convp 0$ since $\hat{\theta}_{NT} \convp \theta_0$. The proof of $S_1 \convp 0$ is similar to the proof of \eqref{clm_dr2_app_eq2}, therefore omitted here. So the proof of Theorem \ref{proposition_wald_test_ptngarch} is completed.

\end{appendix}

\newpage
\bibliographystyle{apalike}
\bibliography{references}

\end{document}